\documentclass[11pt,letterpaper]{article}

\usepackage{latexsym}
\usepackage{epsfig}                     
\usepackage[usenames,dvipsnames]{color}
\usepackage{graphicx}
\graphicspath{{./images/}}
\usepackage{amssymb, amsmath, amsthm}   
\usepackage{enumerate}                  
\usepackage{mathrsfs}                   
\usepackage{multirow}
\usepackage{algorithm}
\usepackage{algorithmic}
\usepackage{longtable}
\usepackage{setspace}
\usepackage{placeins}
\usepackage{ wasysym }
\usepackage{caption}
\usepackage{geometry}                                                                          
\usepackage[dvipsnames]{xcolor} 

\usepackage[T1]{fontenc}
\usepackage{lmodern}

\usepackage{xspace}

\definecolor{purpcol}{RGB}{195,9,235}
\definecolor{citescol}{RGB}{194,101,1}
\definecolor{urlscol}{RGB}{0,150,206}
\definecolor{linkscol}{RGB}{149,0,207}
\definecolor{mycol}{RGB}{25,23,191}
\definecolor{outputcol}{RGB}{34,139,34}
\definecolor{tcol}{RGB}{165,0,14}

\usepackage[breaklinks]{hyperref}
\usepackage[all]{hypcap}
\hypersetup{colorlinks=true,linkcolor=linkscol,citecolor=citescol,urlcolor=urlscol}

\usepackage[normalem]{ulem} 

\newtheorem{thm}{Theorem}
\newtheorem{lem}[thm]{Lemma}

\newtheorem{cor}[thm]{Corollary}

\theoremstyle{definition}

\theoremstyle{definition}
\newtheorem{rem}[thm]{Remark}
\theoremstyle{definition}

\newcommand{\SB}{{\bar{s}}}

\newcommand{\citep}{\cite}
\newcommand{\citet}{\cite}

\newcommand{\tsrev}[1]{{\textcolor{black}{#1}}} 

\begin{document}
\title{The fossilized birth-death model for the analysis of stratigraphic range data under different speciation modes} 
\author{Tanja Stadler$^{1,2}$, Alexandra Gavryushkina$^{1,2}$, Rachel C.\ M.\ Warnock$^{1,2}$,\\ Alexei J.\ Drummond$^3$, Tracy A.\ Heath$^4$}
\date{%
{\it $^1$Department of Biosystems Science \& Engineering, Eidgen\"{o}ssische Technische Hochschule Z\"{u}rich, 4058 Basel, Switzerland;\\
$^2$Swiss Institute of Bioinformatics (SIB), Switzerland.\\
$^3$Centre for Computational Evolution, Department of Computer Science, University of Auckland, Auckland, 1010, New Zealand;\\
$^4$Department of Ecology, Evolution, \& Organismal Biology, Iowa State University, Ames, Iowa, 50011, USA\\ 
\vskip0.5cm}
\today
}

\maketitle

\doublespace

\abstract
A birth-death-sampling model gives rise to phylogenetic trees with samples from the past and the present. Interpreting ``birth'' as branching speciation, ``death'' as extinction, and ``sampling'' as fossil preservation and recovery, this model -- also referred to as the fossilized birth-death (FBD) model -- gives rise to phylogenetic trees on extant and fossil samples. The model has been mathematically analyzed and successfully applied to a range of datasets on  different taxonomic levels, such as  penguins, plants, and insects. However, the current mathematical treatment of this model does not allow for a group of temporally distinct fossil specimens to be assigned to the same species. 

In this paper, we provide a general mathematical FBD modeling framework that explicitly takes ``stratigraphic ranges'' into account, with a stratigraphic range being defined as the lineage interval associated with a single species, ranging through time from the first to the last fossil appearance of the species.
To assign a sequence of fossil samples in the phylogenetic tree to the same species, \textit{i.e.}, to specify a stratigraphic range, we need to   define the mode of speciation. We provide expressions to account for three common speciation modes:  budding (or asymmetric) speciation,  bifurcating (or symmetric) speciation, and anagenetic speciation.

Our equations allow for flexible joint Bayesian analysis of paleontological and neontological data. 
Furthermore, our framework is directly applicable to epidemiology, where a stratigraphic range is the observed duration of infection of a single patient, ``birth'' via budding is transmission, ``death'' is recovery, and ``sampling'' is sequencing the pathogen of a patient. Thus, we present a model that allows for incorporation of multiple  observations through time from a single patient.

\section{Introduction}

Inferring species phylogenies and ultimately the tree of life is one of the main goals of systematics and evolutionary biology. 
Based on inferred species phylogenies, biologists then aim to uncover the dynamics of speciation and extinction (including rates and times). Recovered fossils and sampled extant species are outcomes of a single diversification process of speciation and extinction, and thus share the same evolutionary history. 
Ideally, paleontological and neontological data should be used in combination for reconstructing species phylogenies and estimating speciation and extinction dynamics \cite{Wagner1995, QuentalMarshall2010, SlaterHarmon2013, HuntSlater2016}. 

Joint inference of a time-calibrated phylogeny of living and extinct taxa together with the rates of speciation and extinction requires a model for lineage diversification that gives rise to 
extant species and fossil samples.
Such a model defines the probability density of a rooted phylogenetic tree of extant species and fossils, conditioned on the speciation, extinction, and sampling parameters of the model. 
This probability density then directly allows us to infer the parameters of the model given a phylogenetic tree, using maximum likelihood or Bayesian inference methods. 
Furthermore, based on molecular or morphological information for the extant species and fossil samples, this probability density -- together with models of molecular sequence and morphological character evolution -- allows us to infer the dated phylogeny of observed extant species and fossils \cite{ZhangEtAl2016,GavryushkinaEtAl2017}. \tsrev{This latter}  inference was initially introduced as the total-evidence dating approach  by Ronquist et al. \cite{Ronquist2012}, where the model on the phylogenetic tree was an extension of the uniform prior on \tsrev{ultrametric} clock trees to trees with terminals of different ages, meaning no speciation, extinction, and sampling parameters were specified.

A popular model giving rise to extant species and fossil samples is  the fossilized birth-death (FBD) process \cite{Stadler2010JTB, HeathEtAl2014}. 
The process starts with one lineage (the initial species) at some time in the past. Each lineage has a rate of branching speciation (birth) and a rate of extinction (death). 
Further, each lineage has a rate of producing a fossil sample. 
At the present, each extant lineage has a probability of being sampled. 
The tree displaying all extant and extinct lineages of the FBD process together with the samples is called the complete tree.
Pruning all lineages without sampled descendants from the complete tree gives rise to the ``sampled tree'' 
on extant and extinct samples \cite{Stadler2010JTB}. A sampled tree is a model for a phylogenetic tree inferred from empirical data. 

The sampled tree has a degree-one node at the start of the initial lineage,   degree-three nodes corresponding to branching events, degree-two nodes corresponding to fossil samples being ancestors of other samples, and degree-one nodes corresponding to extant samples or fossil samples without sampled descendants. 
In our terminology a branch in a complete or sampled tree always connects two adjacent degree-one or degree-three nodes, that is, any branching node necessarily terminates a branch and starts two new branches. 
Note that we assume  that degree-two nodes (fossil samples) do not subdivide lineages into branches, unless otherwise stated in the subsection.

Stadler \cite{Stadler2010JTB} provides an example of a complete and a sampled tree in figure 1 of that paper. 
The probability of a sampled tree on extant and fossil samples was calculated in Stadler \cite{Stadler2010JTB} and later in Didier et al.\ \cite{DidierEtAl2012JTB}.
The equations have been implemented in a Bayesian framework for phylogenetic inference as a stand-alone tool \cite{HeathEtAl2014}, 
and as part of the 
BEAST v2.0 \cite{beast2,GavryushkinaEtAl2014,GavryushkinaEtAl2017}, MrBayes \cite{mrbayes, RonquistMrBayes2012, ZhangEtAl2016}, and RevBayes \cite{Hohna2016RevBayes} software packages.
Recently, Didier et al. \cite{DidierEtAl2017} provided a method to evaluate the probability of a sampled tree topology, rather than the sampled tree with branch lengths.

In the FBD model definition, we use the word ``lineage'' rather than ``species''. 
Branching speciation gives rise to an additional species, \textit{i.e.}, co-existing lineages in the sampled tree correspond to different species. 
However, the FBD model does not assign species to lineages through time. 
In particular, a branching speciation event can be considered to occur either via budding (asymmetric) speciation, 
where a single descendant species branches off the ancestral species and both species exist after the speciation event, or via  bifurcating (symmetric) speciation where the ancestral species goes extinct and the event gives rise to two new descendant species (Figure\ \ref{speciation_types}, (i) and (ii)). 
This means that the assignment of species to branches is not specified, and, in particular, branches in the sampled tree do not necessarily correspond to unique species. 
Several  branches in a complete or sampled tree may correspond to the same species due to  budding speciation. 
 For example, consider the tree in  in Figure~\ref{asymmetric_speciation} showing budding speciation: Sp.~1 and Sp.~2 are each represented by 2 branches.
However, a single branch in a sampled tree may also correspond to several different species due to unobserved branching speciation events, \textit{i.e.}, speciation events leading to unsampled lineages. 
For example, in Figure~\ref{asymmetric_speciation}, assume that Sp.~2 is not sampled. Then the branch from the observed speciation event (\textit{i.e.}, the budding event from Sp.~1) to the tip Sp.~3 would represent 2 species, namely Sp.~2 and Sp.~3.
At these unobserved branching speciation events, the species assignment of a branch may change depending on the mode of speciation as defined in Figure\ \ref{speciation_types} (\textit{e.g.}, the species assignment always changes under symmetric speciation).
In summary, while the  FBD model assumes that every co-existing lineage at a particular instant in time belongs to a different species, the FBD model does not make statements about species assignments for lineages through time. 
\begin{figure}
\begin{center}
\includegraphics[width=0.8\textwidth]{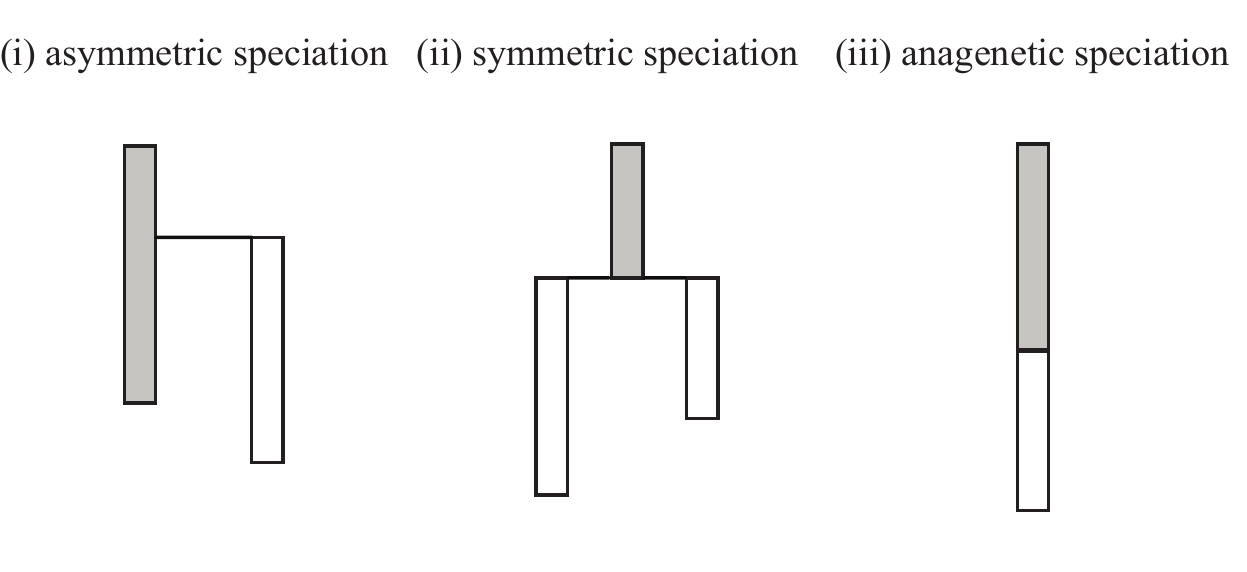} 
\caption{Three speciation modes as described in Foote \cite{Foote1996-fn}. 
The gray and white rectangles represent distinct species. In (i) asymmetric or budding speciation, the ancestral species (gray rectangle) survives after the speciation event whereas in the \tsrev{(ii) symmetric or bifurcating} and (iii) anagenetic cases, the ancestral species is replaced by two or one descendant species.}\label{speciation_types}
\end{center}
\end{figure}

\begin{figure}
\begin{center}
\includegraphics[width=0.7\textwidth]{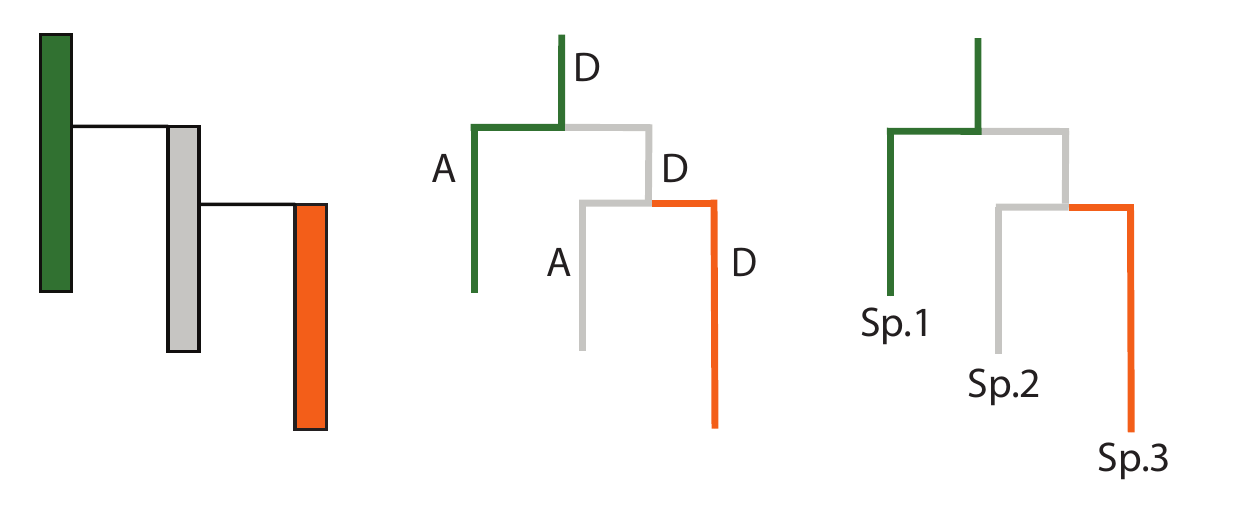} 
\caption{A complete species tree of three species that originated through asymmetric speciation is shown on the left. In the middle, an ``oriented'' species tree is shown with asymmetric speciation corresponding to the species tree of the same three taxa. 
At each speciation event, one of the two new branches is labeled with $A$, because it represents a continuation of the ancestral species, and the other with $D$, designating the new descendant species. 
In an oriented tree, every species is \tsrev{identified} by a unique sequence of $A$ and $D$ branches. Thus, the oldest species \tsrev{is identified by} $DA$, the one that diverges next by $DDA$, and the most recent by $DDD$. 
On the right, a labeled species tree is shown where the orientations are omitted and every species is assigned with a label (taxon name) instead.  
The labeled tree representation is more common for existing phylogenetic software. In all three representations the same-colored segments represent the same species.}\label{asymmetric_speciation}
\end{center}
\end{figure}

In reality, the fossil record often contains multiple observations of the same species over distinct time intervals specifying stratigraphic ranges. 
In order to use this stratigraphic range information from the fossil record, we extend the FBD model to  allow for speciation events of  three different speciation modes
asymmetric (budding), symmetric (bifurcating), and anagenetic (Figure~\ref{speciation_types}). 
Using this extension, we derive the probability density of a sampled tree with stratigraphic ranges. 
As in previous versions of the FBD model, we do not need to assume that we sampled all species, instead the sampling rate explicitly acknowledges incomplete sampling. 
Our equations for the FBD model extension  allow analysis of stratigraphic range data in a phylogenetic framework.

This paper follows a particular structure to  present our new extensions of the FBD model. 
First, we formally define the three speciation modes extending the classic  FBD model in ``The FBD model under three modes of speciation''.
Second, we derive the probability density of a sampled tree on stratigraphic ranges under asymmetric speciation  in Section ``Mathematics of the asymmetric speciation FBD model''. 
Third, we derive the probability density under the three modes of speciation in Section ``Mathematics of the mixed speciation FBD model''.
Based on the mathematical results, we use the section ``Marginalizing over the number of fossils within a stratigraphic range'' to describe the derivation of the probability density of a sampled tree on stratigraphic ranges, given that we only know the time of the first and the last fossil sample, rather than the total number of fossil samples within a stratigraphic range.
In section  ``Marginalizing over the number of fossils within a stratigraphic interval'', we derive the probability density of a sampled tree on stratigraphic ranges, given that we only know whether a fossil species was present or absent within a stratigraphic interval (\textit{i.e.}, a time interval), rather than the total number of fossil samples within each interval. 
We summarize the main results in the discussion, highlighting the conceptual use of such equations in a statistical inference framework. 
Further, we discuss the potential \tsrev{of} these new equations \tsrev{for} contribut\tsrev{ing} to advances in the field of macroevolution.
Finally, we highlight how the equations for the asymmetric speciation case, can be directly employed in molecular epidemiology.

\section{The FBD model under three modes of speciation}
Here we  extend the FBD model towards assigning species to lineages through time.  
For species assignment, we need to specify the speciation mode. We consider three modes of speciation as defined in Foote \citet{Foote1996-fn} (see Figure~\ref{speciation_types}). (i) {\it Asymmetric speciation} where an ancestral species gives rise to a new species via budding, \textit{i.e.}, the descendant species branches off the ancestral species and both species exist after the speciation event. (ii) {\it Symmetric} and (iii)  {\it anagenetic speciation} where the ancestral species goes extinct at the speciation event and gives rise to two (in the symmetric case) and to one (in the anagenetic case) new species. 
Thus, in addition to branching speciation (birth), extinction (death), and sampling events in the FBD model, each branching speciation event is assigned to a mode of speciation (asymmetric or symmetric), and anagenetic changes are marked along lineages. 
Thus, these three modes of speciation events partition all lineages into segments representing distinct species. 
All fossil samples that come from the same segment are assigned to a single species corresponding to that segment. 

A ``stratigraphic range'' defines a continuous lineage between the first and last fossil appearance of a species. 
The FBD model with an assignment of species to lineages through time gives rise to a probability density of a sampled tree on stratigraphic ranges, \textit{i.e.}, on extant  and fossil samples where each sample is assigned to a species. 

Thus, using the FBD model with speciation modes for empirical analysis allows us to assign several fossils to each species in a data analysis.  
At its most conservative interpretation, a stratigraphic range is a single morphospecies observed in multiple stratigraphic layers \cite{Oliver1996morphospp}.
Since morphospecies identifications across stratigraphic intervals are the primary data used in paleontology to study diversity and diversification rates, the FBD model with specification of speciation modes allows us to use the primary paleontological data (stratigraphic range data), jointly with extant species data for phylogenetic analysis.

\section{Mathematics of the asymmetric speciation FBD  model} 

In this section, we formally define the FBD model under asymmetric speciation as illustrated in Figure\ \ref{asymmetric_speciation}. 
As a model for speciation and extinction, we assume a birth-death process with each lineage having a branching speciation rate $\lambda$ and  extinction rate $\mu$. 
The process starts with one lineage at time $x_0$ in the past (also called the origin time) and terminates at the present time 0.  
Table \ref{Table:Notation} gives an overview of these and all other parameter definitions used throughout this paper.

\begin{longtable}{|l|p{5in}|}
  \caption{Notation used throughout this paper.}
  \label{Table:Notation}\\
  \hline
$\lambda$	 & rate of branching speciation\\
$\lambda_a$ & rate of anagenetic speciation\\
$\beta$ & probability of symmetric (vs.\ asymmetric) speciation\\
$\psi$ & fossil sampling rate\\
$\rho$ & extant species sampling probability\\
$\mu$ & extinction rate\\
$\eta$ & $(\lambda, \beta, \lambda_a, \mu, \psi, \rho)$\\
$x_0$ & time of origin of a tree\\
$x_1,\ldots,x_{n-j-1}$ & speciation times in a sampled tree\\
$A/D$  & Orientation of the two branches descending a budding branching event\\
$left/right$  & Orientation of the two branches descending a general branching event\\
$n$ & number of sampled stratigraphic ranges, \textit{i.e.}, number of sampled species (some stratigraphic ranges may only be represented by a single sample in the past or present)\\
$m$ & number of sampled stratigraphic ranges where the associated species goes extinct before present\\
$l$ & number of sampled stratigraphic ranges with an extant species sample\\
$j$ & number of sampled-ancestor-stratigraphic ranges\\
$k$ & total number of sampled fossils\\
$\kappa'$ & total number of sampled fossils that represent the start and end times of  stratigraphic ranges (including ranges represented by a single occurrence) \\
$\kappa$ & total number of sampled fossils within the stratigraphic ranges (\textit{i.e.}, $k=\kappa + \kappa'$)\\
$\kappa_\SB$ & indicates the presence of a fossil within a stratigraphic interval if =1, \tsrev{and} absence if =0\\
$v$ & number of branching speciation events in the labeled tree where we know the orientation\\
$w$ & number of budding speciation events (out of the $n-j-1$ speciation events) in a sampled tree\\
$d_i$ & extinction time of species associated with stratigraphic range $i$\\
$o_i$ & time of first observed fossil corresponding to the species represented by stratigraphic range $i$, \textit{i.e.}, start time of stratigraphic range $i$\\
$y_i$ & time of last observed fossil corresponding to the species represented by  stratigraphic range $i$, \textit{i.e.}, end time of stratigraphic range $i$\\
$b_i$ & branching event in extended sampled tree giving rise \tsrev{to the straight line on which stratigraphic range $i$ lies, also called birth time of $i$} \\ 
$\gamma_i$ & number of lineages co-existing at the birth time $b_i$\\
$a(i)$ & most recent stratigraphic range ancestral to stratigraphic range $i$\\
$t_i$ & time of augmented unobserved speciation event \tsrev{that gave rise to the species associated with stratigraphic range $i$, meaning $t_i$ is speciation time of that species} \\
$I$ & \tsrev{set of stratigraphic ranges, with $i \in I$ if $i$ is in the same straight line as its most recent ancestral stratigraphic range $a(i)$ in the graphical representation of the sampled tree}\\
$L_s$ & sum of all stratigraphic range lengths\\
$L_\SB$ & length of a sub-branch spanning a stratigraphic interval\\
$s_i$ & start time of branch $i$\\
$e_i$ & end time of branch $i$\\
$\mathcal{T}_e^o$ & oriented extended sampled tree\\
$\mathcal{T}_e^l$ & labeled extended sampled tree\\
$\mathcal{T}_s^o$ & oriented sampled tree\\
$\mathcal{T}_r$ & tree when ignoring the $\kappa$ fossils within stratigraphic ranges\\
$\mathcal{T}_l$ & tree when ignoring the number of fossils within a stratigraphic interval\\
$\mathcal{D}$ & \tsrev{Summary of fossil occurrence data with $k$ sampled fossils, $l$ sampled extant species, and $n$ sampled stratigraphic ranges with times $o_i, b_i, d_i$}\\
\hline
\end{longtable}

To model ``asymmetric'' or ``budding'' speciation, we assume that one of the two descendant branches of a speciation event belongs to  the ``ancestral'' species and the other belongs to the ``descendant'' species, thus the two descendant branches may be assigned label $A$ for ancestral and label $D$ for descendant species (Figure~\ref{asymmetric_speciation}). 
Following Ford et al.\ \cite{FordEtAl2009}, we call the label assignment for each pair of descendant branches of a speciation event an ``orientation''. 
In an {\it oriented tree} every species is represented by a path that starts with a $D$-branch and may be continued by several (or none) $A$-branches.  
For example, species 2 in the middle tree in Figure\ \ref{asymmetric_speciation} comprises two branches: the initial $D$-branch and one $A$-branch, because it is ancestral to another species (species 3). Species 3 consists of only the starting $D$-branch, because it does not give rise to any other species.

Typically, the graphical representation of trees used in the paleontological literature (\textit{e.g.}, \cite{Bapst2013}) draws all $A$- and $D$-branches that belong to the same species in a single straight line (Figure~\ref{asymmetric_speciation}, left). 
This graphical representation implicitly contains the information introduced by the $A$ and $D$ orientation (Figure~\ref{asymmetric_speciation}, middle). 
Therefore in the remaining figures of the text we use this graphical representation and omit reference to the $A$ and $D$ orientation. 

In addition to the birth-death process with species assignment, we assume a sampling process for fossils and extant tips. 
We assume fossil sampling occurs along each lineage with rate $\psi$, and an extant species is sampled with probability $\rho$. 
The tree that includes all extant and extinct species that evolved from a single ancestor during time interval $x_0$ together with all the samples is called the ``complete tree''. 
Figure \ref{FigTree}, left (ignoring the blue and grey colors for the moment), displays a complete tree on  eight species with the extant and fossil samples shown using black diamonds. 
Five species have sampled fossils (species 1,2,4,5,6), and two species (species 3,4) have an extant sample. 
Two species (7 and 8) are not sampled.

\begin{figure}
\includegraphics[width=15cm]{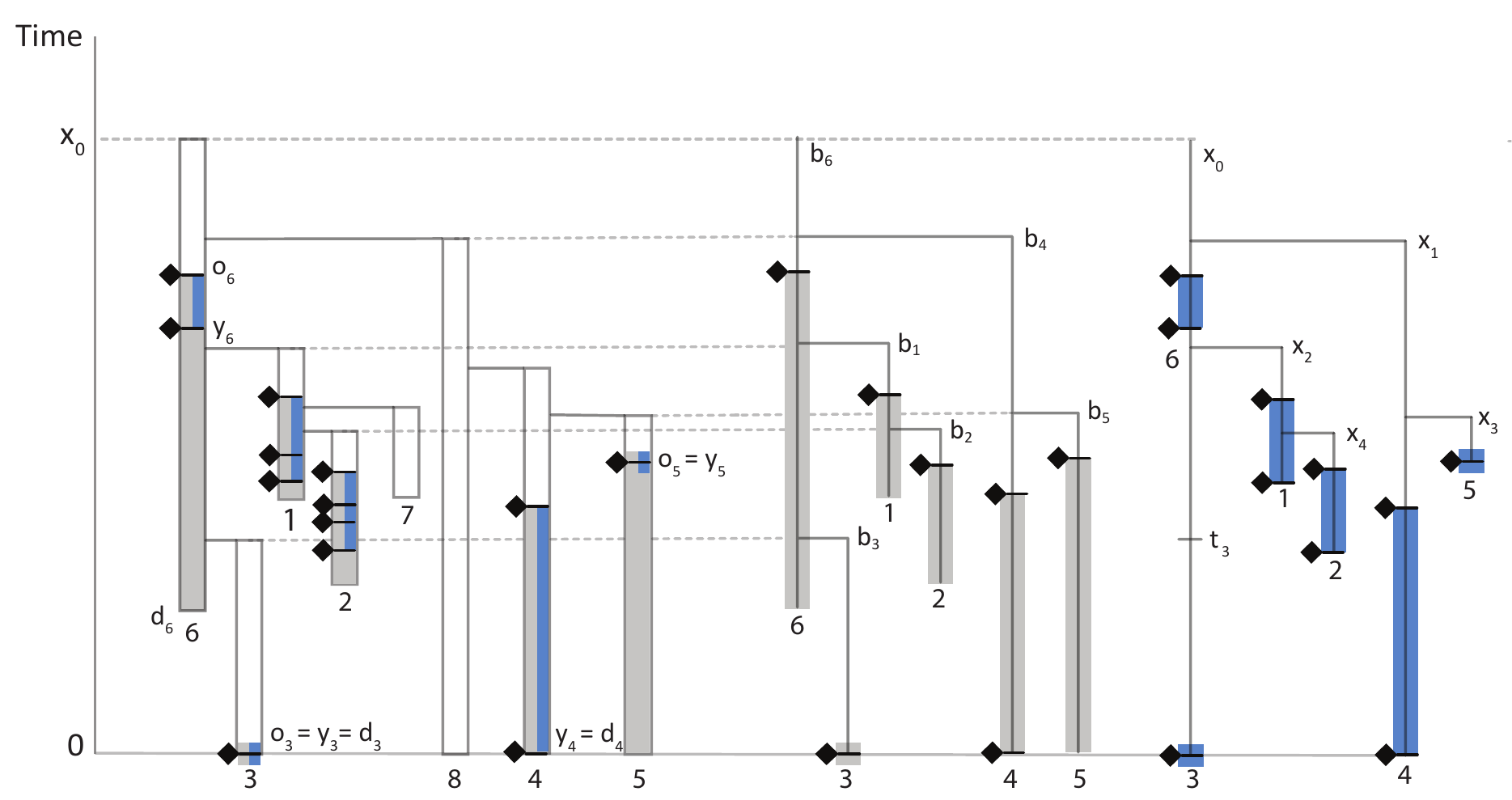}
\caption{Example of a complete tree (left) and its extended sampled tree (middle) and sampled tree (right). 
We mark all fossil and extant species' samples with a diamond. 
The stratigraphic ranges are marked in blue, the extended stratigraphic ranges in grey. 
 We remind the reader that a  straight line in these trees represents our graphical representation,  meaning the oldest branch in a line is the $D$ branch and all subtending branches are $A$ branches. We omit $D/A$ here for clarity of the figure. Furthermore, we omit in the extended sampled tree the fossils within stratigraphic range $i$ that are younger than $o_i$, and in the sampled trees the fossils that appear between $o_i$ and $y_i$, as \tsrev{the times of} these fossils do not contribute to the probability density of the respective tree.
The numbering of species and bifurcation events is chosen to simplify the notation and does not reflect the chronological order of the events.
Theorem \ref{ThmExtendedStickTree} provides the probability density for the oriented extended sampled tree, Corollary \ref{CorExtendedStickTreeLabeled}  for the labeled extended sampled tree, and Corollary \ref{CorExtendedStickTreeCollapse} for the  extended sampled tree when summing over the possible tree topologies. Theorem \ref{ThmStickTree} provides the probability density for the oriented sampled tree.} \label{FigTree}
\end{figure}

\subsection{Stratigraphic ranges and extended stratigraphic ranges}

We now assign all samples belonging to the same species to a stratigraphic range, where $o_i$ (time of the oldest sample of \tsrev{stratigraphic range} $i$) and $y_i$ (time of the youngest sample of \tsrev{stratigraphic range} $i$) are the first and last sampled appearances, respectively.
Note that a stratigraphic range is a segment of a lineage that does not contain $D$-branches, with the exception of the first branch belonging to the segment. 
In other words, it is simply a segment of a straight line in the graphical representation of the tree.
The stratigraphic ranges in Figure\ \ref{FigTree} are marked in blue. 
For all species where only one sample is collected, we have $o_i=y_i$.
For species where we only have an extant sample, the stratigraphic range is only represented by that particular extant sample (species 3 in Figure\ \ref{FigTree}); 
for species where we only have one fossil sample and no extant samples, the stratigraphic range is only represented by that particular fossil sample (species 5 in Figure\ \ref{FigTree}).

We denote the extinction time of \tsrev{ a species associated with stratigraphic range} $i$ with $d_i$ ($d$ for death). We set $d_i=0$ for \tsrev{the associated species being} an extant species. 
An ``extended stratigraphic range'' defines a continuous lineage between the first appearance of a species (time $o_i$ for \tsrev{stratigraphic range} $i$) and the extinction of the species (time $d_i$ for \tsrev{stratigraphic range}  $i$).
The extended stratigraphic range for the six sampled species in Figure~\ref{FigTree} are highlighted in grey. We have $d_i=0$ in the case of the sampled species surviving to the present (species 3,4 in Figure~\ref{FigTree}).
If $y_i=0$, the extended stratigraphic range is equivalent to the stratigraphic range (species 3 and 4 in Figure~\ref{FigTree}).
 
Note that species 6 in Figure\ \ref{FigTree} has values $o_6, y_6, d_6$ displayed, but we dropped these values for some of the other species for clarity of the figures.

\subsection{Sampled trees and extended sampled trees}
Lineages without sampled descendants are deleted from the complete tree to obtain a sampled tree.
In this section, we discuss two types of sampled trees: the ``sampled tree'' and the ``extended sampled tree'' (see Figure\ \ref{FigTree}).
To obtain the ``sampled tree'' (or phylogenetic tree) on stratigraphic ranges, all lineages without sampled descendants are deleted from the complete tree. 
Each branching event that is maintained in the extended sampled tree inherits the labels $A$ and $D$ from the corresponding branching event in the complete tree. 
As above, going from root to tip, at each speciation event we draw an $A$-branch as a straight line directly below its ancestral $D$-branch, while we draw the $D$-branch to the right of its ancestral branch (Figure~\ref{FigTree}, see left for the complete tree, and right for the sampled tree).
Thus we generalize the graphical representation to sampled trees, where some species may be missing, meaning straight lines may now correspond to several species. A stratigraphic range remains as a segment of a straight line in the graphical representation of the sampled tree.

To obtain the ``extended sampled tree'' on extended stratigraphic ranges, we delete lineages with no descendant samples from the complete tree keeping the lineages leading to the extinction times $d_i$ of each sampled species $i$ (Figure~\ref{FigTree}, see left for the complete tree, and middle for the extended sampled tree).

Note that the extended sampled tree and the sampled tree are oriented trees.  
Oriented trees facilitate derivations of probability densities, while most phylogenetic trees inferred from empirical data are ``labeled'' trees, \textit{i.e.}, trees where each sample has a unique label but no orientations are assigned.
We obtain a labeled tree from an oriented tree by omitting all $A$/$D$ orientations, and labelling the sampled species uniformly at random  with unique labels. 
However, despite ignoring the orientation in a labeled tree, we may still know the ancestor-descendant relationships of some branches in the extended sampled tree or in the sampled tree, namely  a new species  budding off from a stratigraphic range is known to be the descendant  (\textit{e.g.}, species $2$  in Figure~\ref{FigTree} is a descendant). 
 We will show below how to transform the probability density of an extended sampled tree into a probability density of a labeled extended sampled tree, such that our results can be applied to labeled trees. For sampled trees, we could not find such a transformation.

Further, note that typically we have information about the stratigraphic range only, and not the extended stratigraphic range, since we do not know the extinction time $d_i$ for each stratigraphic range $i$. 
Nevertheless, marginalizing over unknown $d_i$ using numerical techniques (\textit{e.g.}, \tsrev{Markov chain Monte Carlo}, MCMC) may be advantageous compared to considering sampled trees when we have stratigraphic range data but no information about their phylogenetic relationships. 
In this case, we cannot infer the underlying tree topology, however, using extended stratigraphic ranges, we can integrate over tree topologies analytically (Section ``Probability density of the extended sampled stratigraphic ranges'' below).

In what follows, first, we calculate the probability density of an extended sampled tree, 
including expressions for when we only know the extended stratigraphic ranges but lack information on the phylogenetic relationship of these ranges (thus all tree topologies which have the range data embedded may be possible).
Using results for the extended sampled tree, we then calculate the probability density of a sampled tree.

\subsection{Probability density of the extended sampled tree}

We now calculate the probability density of an extended sampled tree. This allows us to estimate the parameters of the FBD  model under asymmetric speciation, $\lambda, \mu, \psi, \rho$, from the extended sampled tree. 
Additionally, this probability density can be used as the tree prior in Bayesian inference to estimate the extended sampled tree, given observed stratigraphic range data. 

For the derivation of the probability density, we need some notation. Throughout this paper, $n$ is the number of sampled species, \textit{i.e.}, in the context of extended sampled trees, $n$ is the number of extended stratigraphic ranges.
The extended sampled tree on $n$ sampled species is a binary tree with $n-1$ branching events. 
One stratigraphic range $i$ has birth time $b_i=x_0$, meaning the stratigraphic range  did not originate via speciation but started the process (species 6 in Figure\ \ref{FigTree}).
All other stratigraphic ranges originate via branching off another lineage, or, more formally, via the branching event giving rise to the most recent $D$-branch ancestral to that particular stratigraphic range.  
For a given stratigraphic range $i$, this time is denoted by $b_i$ ($b$ for birth).  
Note that $b_i$  is the speciation time of the \tsrev{species associated with stratigraphic range} $i$ in the case of a fully sampled tree, but may be the speciation time of a non-sampled ancestor of the sampled stratigraphic range $i$ in the case of incomplete sampling (such as species 4 in Figure\ \ref{FigTree}). 

Let $k$ be the total number of sampled fossils
and $m$ represents the number of sampled  species going extinct before the present, where $d_i>0$.
Of the $n-m$ number of stratigraphic ranges with $d_i = 0$, let $l$ stratigraphic ranges have an extant sample.
In Figure\ \ref{FigTree}, $n=6,k=11,m=3$ (species 1, 2 and 6 go extinct), and $l=2$ (species 3 and 4 have an extant sample). 

Note that if $\mu=0$ and $\rho=1$ we sample all species, as we have no extinct species and sample every extant species.
We call this case the ``guaranteed complete sampling'' case. All other parameter combinations are referred to as ``potential incomplete sampling''.

\begin{thm} 
\label{ThmExtendedStickTreeComplete}
In the case of guaranteed complete sampling ({\em i.e.}, $\mu=0$ and $\rho=1$), the probability density of the oriented extended sampled tree, $\mathcal{T}^o_e$, is,
$$f[\mathcal{T}^o_e \mid \lambda, \mu, \psi, \rho, x_{0} ] = \psi^k \lambda^{n-1} \prod_{i=1}^n e^{-(\lambda+\psi)b_i}.$$
\end{thm}

\begin{proof}
We know that \tsrev{a lineage from $b_i$ to the extended stratigraphic range $i$ is associated with a single} species, as we have no unobserved branching events. 
Further, $d_i=0$, as we have no extinction events. 
The probability of no event happening on a lineage for time $b_i$ is $e^{-(\lambda+\psi)b_i}$. 
The rate for each of the $n-1$ speciation events is $\lambda$, the rate for each of the $k$ fossilization events is $\psi$. 
Multiplying these components establishes the theorem.
\end{proof}

We note that it follows from  the last theorem that, under guaranteed complete sampling, the probability density of an extended sampled tree (Figure\ \ref{FigTree}, middle, with $k$ sampled fossils), only depends on $k,n$ and the birth times $b_i$ for each stratigraphic range \tsrev{i}. In the following theorem, we show that under potential incomplete sampling, the probability density of the extended sampled tree in fact only depends on $k,n$ and the  times $b_i,o_i$ and $d_i$ for each stratigraphic range and not on the times of each of the $k$ fossils.

\begin{thm} 
\label{ThmExtendedStickTree}
In the case of potential incomplete sampling, the probability density of the oriented extended sampled tree, $\mathcal{T}^o_e$, is,
\begin{equation} \label{EqnExtendedStickTree}
f[\mathcal{T}^o_e \mid \lambda, \mu, \psi, \rho, x_{0} ] =  \frac{\psi^k \mu^m \rho^l (1-\rho)^{n-m-l}}{\lambda(1-p(x_0))}  \prod_{i=1}^n \lambda  \frac{\widetilde{q}_{asym}(o_i)}{\widetilde{q}_{asym}(d_i)} \frac{q(b_i)}{q(o_i)}
\end{equation}
with,
\begin{eqnarray}
p(t) &=& 1+ \frac{-(\lambda - \mu - \psi) + c_1 \frac{e^{-c_1 t}(1-c_2)-(1+c_2) }{e^{-c_1 t} (1-c_2)+(1+c_2)}} {2 \lambda}, \label{Eqnp} \\
\widetilde{q}_{asym}(t) &:=& \sqrt{e^{-t(\lambda+\mu+\psi)} q(t)}, \label{Eqnqtilde} \\
q(t) &=& \frac{4 e^{-c_1 t}}{(e^{-c_1 t} (1-c_2) + (1+c_2))^2} \label{Eqnq}\\
c_1 &=& |\ \sqrt{(\lambda - \mu - \psi)^2 + 4\lambda\psi}\ |, \notag \\
c_2 &=& - \frac{\lambda - \mu - 2\lambda\rho-\psi}{c_1} \label{Eqnc2}. 
\end{eqnarray}
\end{thm}

\begin{proof}
This proof follows in logic the derivations in Stadler \cite{Stadler2010JTB}. First, define $p(t)$ to be the probability that an individual at time $t$ in the past does not leave any sampled fossils or sampled extant descendants. 
We note that at time $t +\Delta t$ in the past, within a small time interval $\Delta t$,   an extinction event  of a lineage happens with probability ($\mu \Delta t$), no event happens with probability  $(1-(\lambda+\mu+\psi) \Delta t)$, and a speciation event happens with probability $\lambda \Delta t$, thus,
$$p(t+\Delta t) = \mu \Delta t + (1-(\lambda+\mu+\psi) \Delta t)p(t) + \lambda \Delta t p(t)^2.$$
Rearranging and letting $\Delta t \rightarrow 0$ leads to,
$$\frac{d}{dt} p(t) = \mu - (\lambda+\mu+\psi) p(t) + \lambda p(t)^2.$$ 
The initial value is $p(0) = 1-\rho$.
Differentiation of Eq.\,\eqref{Eqnp} and plugging it into this differential equation establishes the expression for $p(t)$ (this was derived in our earlier work  \cite{Stadler2010JTB}).

The probability density of an individual associated with stratigraphic range $i$ at time $t \in [o_i,d_i)$ producing an extended stratigraphic range as observed within $(t,d_i]$ is described by the differential equation,
$$\frac{d}{dt} \widetilde{Q}_{asym}(t) =  - (\lambda+\mu+\psi) \widetilde{Q}_{asym}(t) + \lambda \widetilde{Q}_{asym}(t) p(t).$$
This differential equation is derived analogous to the differential equation for $p(t)$.
The initial value for stratigraphic range $i$ is  $\widetilde{Q}_{asym}(d_i) = c$, where $c=\mu$ if $d_i>0$, $c=\rho$ if $d_i=0$ and the extant species is sampled, and $c=1-\rho$ if $d_i=0$ and the extant species is not sampled.
Differentiation of Eq.\,\eqref{Eqnqtilde} and plugging it into the differential equation shows that $\widetilde{q}_{asym}(t) $ is a solution of the differential equation, with $\widetilde{q}_{asym}(0) =1 $.
Thus for stratigraphic range $i$, $\widetilde{Q}_{asym}(o_i) = \frac{\widetilde{q}_{asym}(o_i)}{\widetilde{q}_{asym}(d_i)} c$. 

Next, stratigraphic range $i$ is traced back into the past from $o_i$ to $b_i$ during which time we do not know if the lineage  belongs to the same species, as there may be unobserved speciation events.
During that interval $[b_i,\tsrev{o}_i)$ the probability density of an individual at time $t$ with $b_i \geq t > o_i$ producing an extended stratigraphic range $i$ as observed is described by the differential equation,
$$\frac{d}{dt} Q(t) =  - (\lambda+\mu+\psi) Q(t) + 2 \lambda Q(t) p(t).$$ 
The initial value for stratigraphic range $i$ is $Q(o_i) = \widetilde{Q}_{asym}(o_i)$.
Note that the additional 2 in the differential equation for $Q(t)$ compared to $\widetilde{Q}_{asym}(t)$ allows for unobserved speciation events where the ancestor species or the descendant species are not sampled, while $\widetilde{Q}_{asym}$ only considers unobserved events where the descendant species is not sampled.
The differential equation for $Q(t)$ has been already solved in Stadler \cite{Stadler2010JTB}: our expression for $q(t)$ in Eq.\ \eqref{Eqnq} is a solution of the differential equation for $Q(t)$ with initial value $q(0)=1$.
Analogous to $\widetilde{Q}_{asym}(o_i)$, we write $Q(b_i) = \frac{q(b_i)}{q(o_i)} \widetilde{Q}\tsrev{_{asym}}(o_i).$

The rate of a lineage originating via budding from another lineage is $\lambda$ 
and sampling of each one of the $k$ fossils happens with rate $\psi$.
Multiplying the probability densities $Q(\tsrev{b}_i)$ for all extended stratigraphic ranges $i$, the speciation rates, and sampling rates, and dividing by the probability of obtaining a sample (\textit{i.e.}, conditioning on sampling via $(1-p(x_0))$), establishes the theorem.
\end{proof}

We note that in the case of guaranteed complete sampling, where $\mu=0$ and $\rho=1$, we have $p(t)=0$ 
and the expression for $\widetilde{q}_{asym}$ simplifies to $\widetilde{q}_{asym}(t)=e^{-(\lambda+\psi)t}$,
an expression that we encountered already in Theorem \ref{ThmExtendedStickTreeComplete}.

\begin{rem}
If we were to know the oriented, complete tree with the fossil samples and extant species samples, meaning there are no unobserved events, regardless of what fossil or extant samples were collected, then we could calculate the probability density of an oriented, complete tree $\mathcal{T}^o_c$ as,  
\begin{equation*} 
f[\mathcal{T}^o_c \mid \lambda, \mu, \psi, \rho, x_{0} ] = \frac{\psi^k \mu^m  \rho^l (1-\rho)^{n-m-l}}{\lambda(1-p(x_0))}  \prod_{i=1}^n \lambda e^{-(\lambda+\mu+\psi)(b_i-d_i)},
\end{equation*}
where $e^{-(\lambda+\mu+\psi)(b_i-d_i)}$ is the probability of observing a single species in the time interval $(b_i,d_i)$.
\end{rem}

\begin{rem}
As a theoretical side note, we further conclude $\frac{q(b_i)}{q(d_i)} \geq \frac{\widetilde{q}_{asym}(b_i)}{\widetilde{q}_{asym}(d_i)} \geq e^{-(\lambda+\mu+\psi)(b_i-d_i)}$. 
For establishing $\frac{\widetilde{q}_{asym}(b_i)}{\widetilde{q}_{asym}(d_i)} \geq e^{-(\lambda+\mu+\psi)(b_i-d_i)}$, we note that the left hand side is the probability density of a given stratigraphic range, with any number of hidden speciation events (including no hidden events); 
the right hand side is the probability density of the stratigraphic range, without hidden speciation events -- this is a special case of the left hand side. 
For establishing $\frac{q(b_i)}{q(d_i)} \geq \frac{\widetilde{q}_{asym}(b_i)}{\widetilde{q}_{asym}(d_i)}$, we note that the right hand side is the probability  density of a stratigraphic range, meaning the lineage between $b_i$ and $d_i$ belongs to the same species, while the left hand side is the probability of a lineage allowing for unobserved speciation events, thus the lineage may correspond to different species before and after unobserved speciation events. 
Again, the right hand side is a special case of the left hand side. 
\end{rem}

Rather than oriented trees, most software packages perform inference over labeled trees (see Figure\,\ref{asymmetric_speciation}, right). 
That means all $n$ sampled species  are labeled uniformly at random with $n$ labels ($n!$ possibilities), and the orientations $A$ and $D$ are summed over, unless we know the orientation. 
We know the orientation if a stratigraphic range produces a new descendant species:  $A$ is the label of the  descending branch associated with the stratigraphic range. 
We denote with $v$ the number of branching speciation events where we know the orientation. 
This leads to the following corollary.

\begin{cor} 
\label{CorExtendedStickTreeLabeled}
In the case of potential incomplete sampling, the probability density of the labeled extended sampled tree, $\mathcal{T}^l_e$, is,
\begin{equation} \label{EqnExtendedStickTreeLabeled}
f[\mathcal{T}^l_e \mid \lambda, \mu, \psi, \rho, x_{0} ] =  \frac{2^{n-v-1}}{n! } \frac{ \psi^k \mu^m \rho^l (1-\rho)^{n-m-l}}{\lambda(1-p(x_0))}  \prod_{i=1}^n  \lambda  \frac{\widetilde{q}_{asym}(o_i)}{\widetilde{q}_{asym}(d_i)} \frac{q(b_i)}{q(o_i)}.
\end{equation}
In the case of guaranteed complete sampling ({\em i.e.}, $\mu=0$ and $\rho=1$), the probability density of the labeled extended sampled tree, $\mathcal{T}^l_e$, is,
$$f[\mathcal{T}^l_e \mid \lambda, \mu, \psi, \rho, x_{0} ] = \frac{2^{n-v-1}}{n!} \psi^k \lambda^{n-1}  \prod_{i=1}^n e^{-(\lambda+\psi)b_i}.$$
\end{cor}

\subsection{Probability density of the extended sampled stratigraphic ranges} 

Next we assume that instead of an extended sampled tree, we only know the $n$ sampled stratigraphic ranges with start times $o_i$ and end times $y_i$ ($i=1,\ldots,n$), of which $l$ contain a sampled extant species, and there are $k$ sampled fossils. 
For each stratigraphic range, we augment our data with the values $b_i > o_i$ and $d_i < y_i$ such that there are no gaps, that is, for each $i=1,\ldots,n$ there is $j$ such that $b_i \in (b_j, d_j)$, \tsrev{with the exception of $i$ for which $b_i=x_0$}. 

 We aim to  calculate the probability density of these stratigraphic ranges with the corresponding $b_i$ and $d_i$.
Using  this probability density, one can estimate speciation and extinction rates based on fossil occurrence data (\textit{i.e.}, stratigraphic ranges) by marginalizing numerically over all possible speciation and extinction times ($b_i, d_i$) using methods such as MCMC. 
This has been done previously assuming all extinct species have at least one fossil sample in Silvestro et al.\ \cite{SilvestroEtAl2014}.

In summary, given $o_i,\ y_i,\ b_i$ and $d_i$ ($i=1,\ldots,n$) together with $k$ and $l$ (we summarize $\mathcal{D}=(k,l,\{b_i,d_i,o_i\},i=1,\ldots,n)$), and the parameters $\lambda,\mu,\psi, \rho$, we need to evaluate the probability density of $\mathcal{D}$ given the parameters. 
The probability density of $\mathcal{D}$ is obtained from Theorem \ref{ThmExtendedStickTreeComplete} and \ref{ThmExtendedStickTree} by integrating over all possible tree topologies which have $\mathcal{D}$ embedded.
The following theorem states this probability density. 

Let $\gamma_i$ be the number of lineages co-existing at the birth time $b_i$ of stratigraphic range $i$. For the oldest stratigraphic range $i$ (with birth time $b_i=x_0$), we have $\gamma_i=1$.
In Figure\ \ref{FigTree}, we have $\gamma_1= 2,\gamma_2=4,\gamma_3=4,\gamma_4=1,\gamma_5=3,\gamma_6=1$.

\begin{cor} \label{CorExtendedStickTreeCollapse}
The probability density for $\mathcal{D}$ under potential incomplete sampling is,
\begin{equation} \label{EqnExtendedStickTreeCollapse}
f[\mathcal{D} \mid \lambda, \mu, \psi, \rho, x_{0} ] =  \frac{\psi^k \mu^m \rho^l (1-\rho)^{n-m-l}}{\lambda(1-p(x_0))}  \prod_{i=1}^n \lambda \gamma_i \frac{\widetilde{q}_{asym}(o_i)}{\widetilde{q}_{asym}(d_i)} \frac{q(b_i)}{q(o_i)}.
\end{equation}
The probability density for $\mathcal{D}$ under  guaranteed complete sampling (\textit{i.e.}, $\mu=0$ and $\rho=1$) is,
$$f[\mathcal{D} \mid \lambda, \mu, \psi, \rho, x_{0} ] = \psi^k \lambda^{n-1}  (n-1)! \prod_{i=1}^n e^{(\lambda+\psi)b_i}.$$

\end{cor}
This corollary is a direct consequence of Theorems \ref{ThmExtendedStickTreeComplete} and \ref{ThmExtendedStickTree} by noting that  an extended stratigraphic range $i$ has rate $\lambda \gamma_i $ to be initiated via speciation by one of the $\gamma_i$ coexisting lineages 
(while in  Theorems \ref{ThmExtendedStickTreeComplete}  and \ref{ThmExtendedStickTree} the rate of a branching event along a particular lineage happens with rate $\lambda$). As for the extended sampled tree, the probability density of the extended sampled stratigraphic ranges only depends on $k,n$ and the  times $b_i,d_i$ and $o_i$ for each stratigraphic range and not on the times of each of the $k$ fossils.

The probability density of $\mathcal D$ is obtained by integrating over oriented trees. 
Note that each tree topology giving rise to $\mathcal D$ has a known orientation at each branching event (as we augmented each stratigraphic range $i$ with $b_i$), implying  $v=n-1$.  
Thus, the probability density of $\mathcal D' = (k,l,\{b_i,d_i,o_i\},f, i=1,\ldots,n)$, where $f$ is a mapping of the intervals to some labels, integrated over labeled trees is obtained by multiplying with $\frac{1}{n!}$.

Theorem \ref{ThmExtendedStickTree} for the FBD  model  on extended sampled trees with stratigraphic ranges is the analog of Eq.\,(1) in Gavryushkina et al.\ \cite{GavryushkinaEtAl2014} for the FBD model on sampled trees without stratigraphic ranges, considering fossil phylogenetic relationships explicitly.  Equivalently, Corollary \ref{CorExtendedStickTreeCollapse} is the analog of Eq.\,(1) in Heath et al.\ \cite{HeathEtAl2014} integrating over fossil phylogenetic relationships analytically.

\subsection{Probability density of the sampled tree} 

For the extended sampled tree with stratigraphic ranges described above, we infer extinction times $d_i$ and avoid considering stratigraphic ranges that are sampled ancestors of other stratigraphic ranges.
In this section, we consider the sampled tree spanning the sampled fossils and extant species, without the extinction times $d_i$ (Figure\ \ref{FigTree}, right).
In a sampled tree,  the stratigraphic range $i$ may be a ``tip-stratigraphic range'', meaning the fossil at time $y_i$ is a tip in the sampled tree, or may be a ``sampled-ancestor-stratigraphic range'', meaning the fossil at time $y_i$ has sampled descendants. Species 1-5 correspond to tip-stratigraphic ranges, and species 6 corresponds to a sampled-ancestor-stratigraphic range in Figure\ \ref{FigTree}. Again, as before, we use $n$ to denote the number of sampled species, \textit{i.e.}, the number of stratigraphic ranges. Recall that an extended sampled tree had $n-1$ branching events. Due to the sampled-ancestor-stratigraphic ranges, a sampled tree may have fewer than $n-1$ branching events.

Let $j$ stratigraphic ranges be sampled-ancestor-stratigraphic ranges (in Figure\ \ref{FigTree}, $j=1$).
The sampled tree has branching times $x_1,\ldots,x_{n-j-1}$, and origin time $x_0$. Note that $x_0,x_1,\ldots,x_{n-j-1}$ of a sampled tree is a subset of $b_1,b_2,\ldots,b_n$ of an extended sampled tree.
For derivations only, we consider the oldest and youngest fossils as explicit nodes that subdivide branches in the sampled tree (in contrast to the sections above where a node had degree three or degree one, and in contrast to Stadler \cite{Stadler2010JTB} where all fossils were treated as  nodes in the sampled tree under the classic FBD model without stratigraphic ranges).
Stratigraphic ranges where $o_i=y_i$ are assumed to have a branch between $o_i$ and $y_i$ with  length 0.
The sampled tree then consists of the following nodes:

\begin{itemize}
\item $n-j-1$ degree-three nodes, with the branching times at $x_1,\ldots,x_{n-j-1}$,
\item $n$ degree-two nodes, at the time of the oldest fossils (\textit{i.e.}, the start of a stratigraphic range) $o_1,\ldots,o_n$, 
\item $j$ degree-two nodes, at the sampled-ancestor-stratigraphic range times $y_i$ with $i$ being a sampled-ancestor-stratigraphic range (in our example in Figure\ \ref{FigTree}, $j=1$ and $i=6$), 
\item $n-j$ degree-one nodes (tips), at the tip-stratigraphic range times $y_i$  with $i$ being a tip-stratigraphic range. Of these $n-j$ nodes, $l$ nodes are at time $y_i=0$. For ease of notation in what follows, we label the stratigraphic ranges that are tip-stratigraphic ranges with $i=1,\ldots,n-j$ (in our example in Figure\ \ref{FigTree} stratigraphic ranges 1-5), 
\item  one degree-one node, the origin of the tree at time $x_0$.
\end{itemize}

Each branch connects two nodes which may be of degree one, two, or three.
Thus in the sampled tree, each branch is either fully part of a stratigraphic range, or not at all part of a stratigraphic range. 
A branch belonging fully to a stratigraphic range is called a ``stratigraphic-range branch''. If a stratigraphic-range branch gives rise to a speciation event, precisely one descendant branch is a stratigraphic range-branch.
In total,  $v$ stratigraphic-range branches give rise to a speciation event. In our example (Figure~\ref{FigTree}), $v=1$, as only stratigraphic range 1 gives rise to one additional species (\textit{i.e.}, in the sampled tree, there is only one speciation event that occurs along the sampled stratigraphic range of a given species).
In total the sampled tree has $3n-j-1$ branches. 

Let  $i \in I$ if stratigraphic range $i$ and its most recent ancestral stratigraphic range,  $a(i)$, lie on a straight line in the graphical representation of the sampled tree.
\tsrev{In our example Figure \ref{FigTree}, we have $I=\{3\}$.}
By definition, stratigraphic range $i \in I$ and its most recent ancestral stratigraphic range $a(i)$ belong to different species, thus we need to ensure that there is an unobserved speciation event between $y_{a(i)}$ and $o_i$. We assume that the \tsrev{the species corresponding to stratigraphic range $i$ originated} at time $t_i \in (y_{a(i)},o_i)$, and first augment our data with these times $t_i$ (see $t_3$ in Figure\ \ref{FigTree}, right). Second, we  analytically integrate over $t_i$.
We refer to the oriented, sampled tree as $\mathcal{T}_s^o$.

To obtain the probability density of an oriented sampled tree, we multiply the contribution of each branch in the sampled tree, as in Theorem \ref{ThmExtendedStickTree}. For a branch with start time $s_i$ and end time $e_i$ \tsrev{(forward in time)}, the contribution is $\frac{\widetilde{q}_{asym}(s_i)}{\widetilde{q}_{asym}(e_i)} $ if it is a stratigraphic range-branch, and  $\frac{q(s_i)}{q(e_i)}$ otherwise. We further need to specify the 
initial values at the tips of the tree. 
The initial value of each tip-stratigraphic range with $y_i>0$ is the probability of having no sampled descendants multiplied by the rate of fossil sampling, $\psi p(y_i)$, and the initial value of each tip-stratigraphic range with $y_i=0$ is the probability of sampling an extant species, $\rho$. 

\tsrev{Additionally, we need to correct for the unobserved speciation times of species from~$I$. First, we need to multiply by $\lambda p(t_i)$, --- the probability of a speciation event at the unobserved speciation time $t_i$ and the fact that one of the lineages descending  the speciation event was not sampled. Second, we need to account for the fact that all the branches belonging to the lineage starting at $t_i$ (for a moment we assume that $t_i$ also subdivides a lineage into branches) and ending at $o_i$ are stratigraphic range-branches (as they all belong to the same species associated with stratigraphic range $i$), although we treated them as non-stratigraphic range-branches in the previous paragraph. That means that we first need to multiply by $\frac {q(o_i)}{q(t_i)}$ and then by $\frac {\widetilde{q}_{asym}(t_i)} {\widetilde{q}_{asym}(o_i)}$.}
Thus we obtain the following directly from Theorem \ref{ThmExtendedStickTree}:

\begin{lem} \label{LemStickTree}
In the case of potential incomplete sampling, the probability density of the oriented sampled tree $\mathcal{T}^o_s$ is,
\begin{equation} \label{EqnStickTree}
f[\mathcal{T}^o_s \mid \lambda, \mu, \psi, \rho, x_{0} ] =  \frac{\psi^k \rho^l \lambda^{n-j-1}}{ 1-p(x_0)}  \prod_{i=1}^{3n-j-1}    \hat{q}_{asym}(B_{i}) \prod_{i=1}^{n-j-l} p(y_i) \prod_{i\in I} \lambda p(t_i) \tsrev{\frac {q(o_i)} {q(t_i)} \frac {\widetilde{q}_{asym}(t_i)} {\widetilde{q}_{asym}(o_i)}}
\end{equation}
where the contribution of branch $B_i$ with start time $s_i$ and end time $e_i$ is,
\begin{eqnarray*}
\hat{q}_{asym}(B_i)=
\begin{cases}
\frac{\widetilde{q}_{asym}(s_i)}{\widetilde{q}_{asym}(e_i)} \quad \text{if branch } i \text{ is a stratigraphic range-branch,} \\ 
  \frac{q(s_i)}{q(e_i)} \quad \text{else.}
 \end{cases}
 \end{eqnarray*}
\end{lem}

Rather than augmenting the state space by $t_i, i \in I$, we can integrate analytically over all $t_i$. \tsrev{The integral over $t_i$ is,}
\tsrev{$$ \int  \frac {p(t) \widetilde{q}_{asym}(t)  }{q(t)}dt  =  - \frac 1 \lambda \frac{ \widetilde{q}_{asym}(t)} {q(t)}.$$}
\tsrev{We evaluate this integral over the interval $[o_i, y_{a(i)}]$,} with $a(i)$ being the most recent ancestral stratigraphic range of $i$ as above, and thus obtain:
\begin{thm} \label{ThmStickTree}
In the case of potential incomplete sampling, the probability density of the oriented sampled tree $\mathcal{T}^o_s$ is,
\begin{multline} \label{EqnStickTree}
f[\mathcal{T}^o_s \mid \lambda, \mu, \psi, \rho, x_{0} ] =  \\\frac{\psi^k \rho^l \lambda^{n-j-1}}{ 1-p(x_0)}  \prod_{i=1}^{3n-j-1}    \hat{q}_{asym}(B_{i}) \prod_{i=1}^{n-j-l} p(y_i) \prod_{i \in I} \tsrev{  \Bigl(1 - \frac {q(o_i)} {\widetilde{q}_{asym}(o_i)} \frac{ \widetilde{q}_{asym}(y_{a(i)})} {q(y_{a(i)})}\Bigr)}.
\end{multline}
\end{thm}
\noindent\tsrev{The term within the right product can be written as $ \frac{q(o_i)}{q(y_{a(i)})} \Bigl(\frac{q(y_{a(i)})}{q(o_i)} - \frac {\widetilde{q}_{asym}(y_{a(i)})} {\widetilde{q}_{asym}(o_i)}\Bigr)$. The right bracket is the probability of zero or more unobserved speciation events that change a species along the lineage starting at $y_{a(i)}$ and ending at $o_i$ minus the probability of zero unobserved speciation events that change the species along this lineage.}

For labeled trees, the term that accounts for unobserved speciation events between ancestor-descendant stratigraphic ranges (  \tsrev{$\prod_{i \in I} \lambda p(t_i) \frac {q(o_i)} {\widetilde{q}_{asym}(o_i)} \frac {\widetilde{q}_{asym}(t_i)}{q(t_i)}$ and $\prod_{i \in I} \Bigl(1 - \frac {q(o_i)} {\widetilde{q}_{asym}(o_i)} \frac{ \widetilde{q}_{asym}(y_{a(i)})} {q(y_{a(i)})}\Bigr)$} above) becomes more complex.  A labeled tree does not show the graphical representation, meaning we do not know which ancestor-descendant stratigraphic ranges lie on a straight line and thus we do not know which ranges need to be separated by a speciation event. Instead we need to integrate over the possibilities of these ranges lying on a straight line or not, which is non-trivial. 
Suppose there are two ancestor-descendant stratigraphic ranges that are separated by an observed branching event. 
As we do not know the orientation of this event, there could have been two possible scenarios: 
either the two ranges lie on the same straight line, and thus the separating branching event is a budding speciation event giving rise to an additional new species. In this case we need to enforce an unobserved speciation event between the two ranges such that it is guaranteed that they belong to different species. Alternatively, the observed speciation event causes the two ranges not to lie on the same straight line, then we do not have to force an unobserved speciation event. 

When several speciation events separate a pair of ancestor-descendant stratigraphic ranges or when the same stratigraphic range is the most recent sampled-ancestor of several stratigraphic ranges we could not find a simple expression for the number of different possible scenarios. Thus we cannot provide an expression for the probability density of labeled trees here.

\begin{cor}
In the case of guaranteed complete sampling, the sampled tree equals the extended sampled tree, as $y_i=0$ for all species. Thus the probability densities from Theorem \ref{ThmExtendedStickTreeComplete} and Corollary \ref{CorExtendedStickTreeLabeled} apply.
\end{cor}

\begin{rem}
An expression for the probability density of sampled trees when ignoring tree topology (analogous to Section ``Probability density of the extended sampled stratigraphic ranges'') does not seem to be  straightforward. 
In fact, it seems more straightforward to integrate over tree topologies of sampled trees using MCMC methods. 
Ignoring tree topology can further be achieved by estimating parameters based on the extended sampled stratigraphic ranges and integrating over $d_i$ using MCMC.

The complication can be attributed to the sampled-ancestor-stratigraphic ranges. 
The $\gamma_i$ for the number of possible attachment points of a stratigraphic range $i$ in the extended sampled tree scenario is independent of the placement of the other stratigraphic ranges. 

In the case of the sampled tree, if we ignore the sampled-ancestor-stratigraphic ranges (\textit{i.e.}, replacing them with normal branches in the sampled tree), 
we can sum over tip-stratigraphic range topologies, analogous to the extended sampled tree scenario where we only have tip-stratigraphic ranges. 

However, we then have to additionally account for  the number of placements of the sampled-ancestor-stratigraphic ranges. 
This number does not seem to follow a simple formula. Consider two sampled-ancestor-stratigraphic ranges with range $(4,3)$ (call it $SA_1$) and $(3.5,2)$ (call it $SA_2$) (see Figure\,\ref{FigExpSampledTree}).
Assume there is one tip-stratigraphic range $X$ with $o_X=y_X=2.5$ and $b_X=5$, then there is space for a sampled-ancestor-stratigraphic range on the interval $(5,2.5)$. 
Additionally, assume there is one tip-stratigraphic range $Y$ with $o_Y=y_Y=1.5$, and $b_Y=5.5$, leaving space on the interval $(5.5,1.5)$, and one tip-stratigraphic range $Z$ with $o_Z=y_Z=1$, and $b_Z=3.8$, leaving space on the interval $(3.8,1)$.

Thus both sampled-ancestor-stratigraphic ranges fit on two  lineages ($SA_1$ on lineages leading to $X$ and to $Y$; $SA_2$ on lineages leading to $Y$ and $Z$), and we could be tempted to multiply by $\gamma_i=2$ for both sampled-ancestor-stratigraphic ranges (meaning we would have in total 4 possible sampled trees). 
However, given $SA_1$ is assigned to lineage $Y$ (out of $X$ and $Y$) then $SA_2$ can only be assigned to $Z$. 
On the other hand, if $SA_1$ is assigned to $X$, then $SA_2$ can be assigned to $Y$ or $Z$. 
Meaning the number of choices for each sampled-ancestor-stratigraphic range are not independent of the other sampled-ancestor-stratigraphic ranges (here we have in total three possible sampled trees; see Figure\ \ref{FigExpSampledTree}). 
This non-independence of sampled-ancestor-stratigraphic range placement in a sampled tree makes the analytic integration over tree topologies non-trivial.

\begin{figure}
\begin{center}
\includegraphics[width=0.6\textwidth]{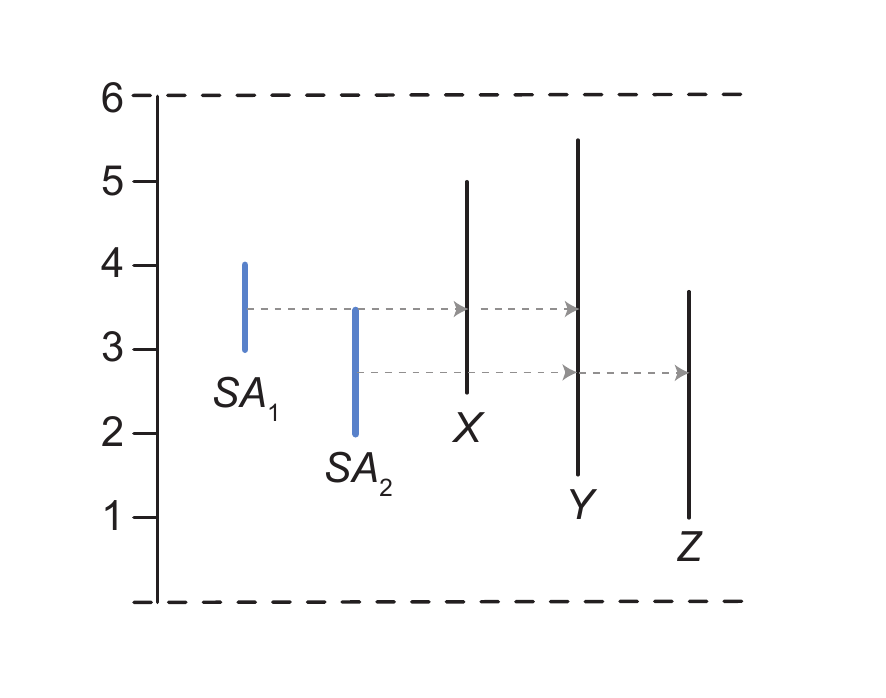} 
\caption{Illustration of sampled-ancestor-stratigraphic range assignment to non-stratigraphic range lineages $X,Y,Z$ of a sampled tree.
If sampled ancestor $SA_1$ is assigned to lineage $X$, then $SA_2$ can be assigned to $Y$ or $Z$, while if $SA_1$ is assigned to lineage $Y$, then $SA_2$ can be assigned only to $Z$. }
\label{FigExpSampledTree}
\end{center}
\end{figure}

If all sampled-ancestor-stratigraphic ranges have length 0 (\textit{i.e.}, $y_i=o_i$), we can analytically sum over topologies following Heath et al.\ \cite{HeathEtAl2014}, as the different sampled ancestor fossils do not influence each other when being assigned to branches.
\end{rem}

\section{Mathematics of the mixed speciation FBD  model} 

After having discussed the FBD  model under asymmetric speciation, we now allow for three speciation modes: 
asymmetric, symmetric, and anagenetic speciation. First we assume that the probability of  a branching speciation event  being symmetric is $\beta$.  That is, we extend the  FBD  model with rates $\lambda, \mu, \psi, \rho$ assigning to each branching event an asymmetric speciation event with probability $1-\beta$ and a symmetric speciation with probability  $\beta$. 
Further, each lineage has rate $\lambda_a$ of producing an anagenetic speciation event, \textit{i.e.}, a speciation event without branching. The mixed speciation model has parameters $\lambda, \mu, \psi, \rho, \beta, \lambda_a$ and setting the additional parameters $\beta$ and $\lambda_a$ to zero converts to the initial asymmetric speciation FBD model.

This mixed speciation FBD model induces oriented trees where each branch is labeled either {\it left} or {\it right}. 
A complete tree produced by this process will be represented by an oriented tree where all nodes have degree-three at most, and all degree-three nodes are of two types reflecting the mode of speciation: asymmetric or symmetric speciation nodes.
At nodes representing an asymmetric speciation event, we always assume that the new species starts with the right branch (which would correspond to a $D$-branch in the previous sections and the left branch would correspond to an $A$-branch). 
The left and right descendant branches of a symmetric speciation event are equivalent and we need the orientation only for the convenience of derivations. 
Further, we have degree-two nodes that represent anagenetic speciation events that also subdivide branches. 
A branch that descends from an anagenetic speciation event has the same orientation as its ancestor branch.

A species in the complete tree under the mixed speciation process is represented by a lineage consisting of a starting branch, which can be: 
\begin{itemize}
\item the initial branch starting at time $x_0$, 
\item a branch produced by symmetric or anagenetic speciation, or
\item the right branch (analogous to the $D$-branch) of an asymmetric speciation event, 
\end{itemize}
and several (or none) left descending branches (analogous to the $A$-branches) produced by asymmetric speciation. 
We  define stratigraphic ranges in the complete tree as before.
 
Following the previous sections, we draw the branches belonging to the same species as straight lines in the complete tree (Figure\,\ref{mixed_speciation}, left). 
Thus, at an asymmetric speciation event the {\it left} branch (analogous to the $A$-branch) continues the ancestral branch and the {\it right} branch (analogous to the $D$-branch) is drawn on the righthand side of the ancestral branch. 
For the symmetric speciation event both descendant branches correspond to a new species and are drawn on both sides of the ancestral branch. 
To designate an anagenetic speciation event we draw the descendant branch slightly shifted to the right of the ancestral branch in the complete tree.

\begin{figure}
\begin{center}
\includegraphics[width=12cm]{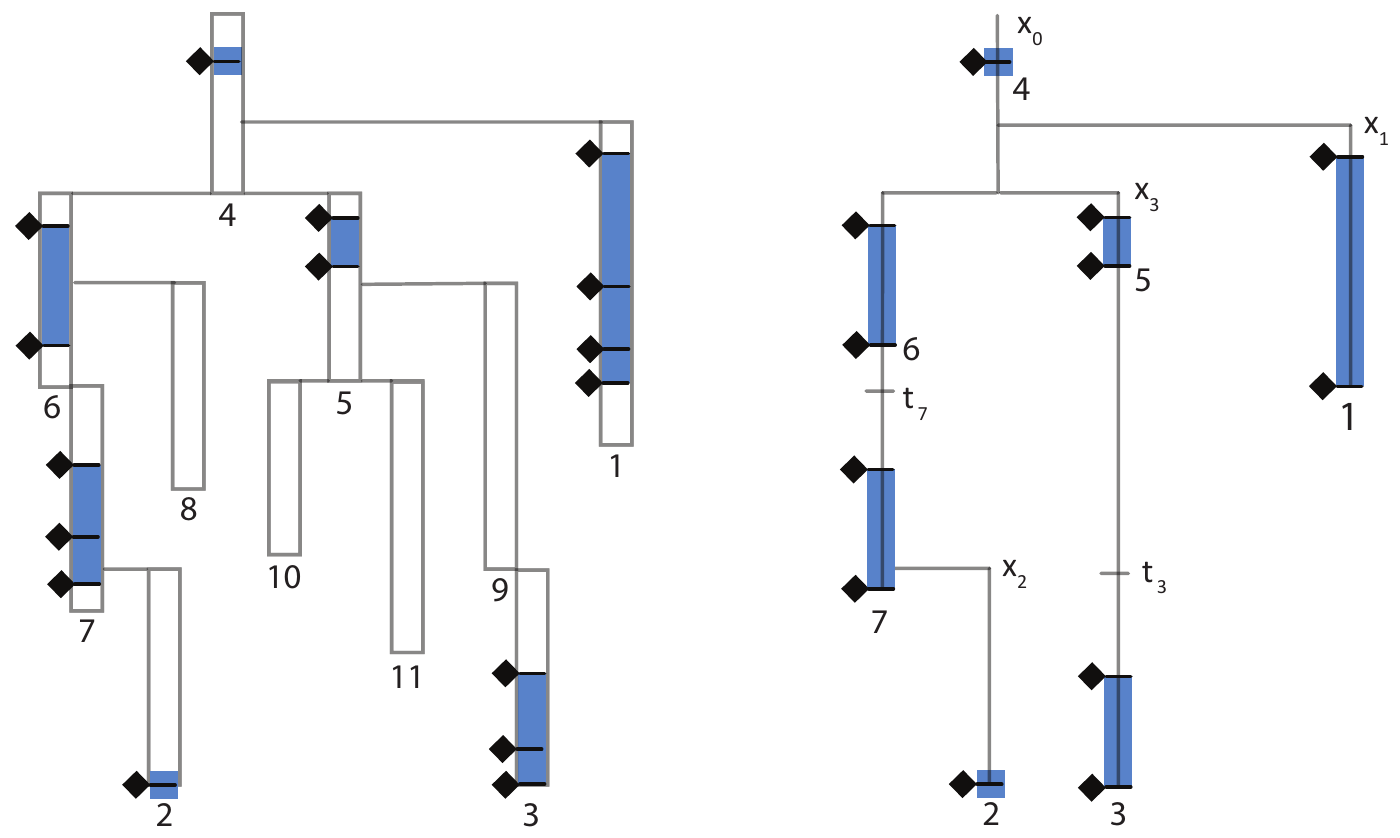} 
\caption{A complete species tree with three 
speciation modes (mixed speciation) is shown on the left. A sampled tree with mixed speciation is shown on the right.
} 
\label{mixed_speciation}
\end{center}
\end{figure}

A sampled tree (Figure\,\ref{mixed_speciation}, right) is obtained by deleting all lineages without sampled descendants and ignoring anagenetic speciation nodes. 
Each branching node inherits its type (asymmetric or symmetric) from the complete tree. 
We draw branches produced by asymmetric and symmetric speciation nodes in the same way as in the complete tree. 
Finally, as in the asymmetric speciation case, a straight line in the sampled tree does not necessarily represent a single species in the sampled tree, as there may be unobserved speciation events.

Analogous to the asymmetric case, a stratigraphic range in the sampled tree is a segment of a lineage that does not contain unobserved symmetric speciation events and \tsrev{asymmetric speciation events where the species associated with the lineage changes}. In other words, it is simply a segment of a straight line in the graphical representation of the sampled tree (see Figure\,\ref{mixed_speciation}).

Suppose again we have sampled $n$ species, \textit{i.e.} $n$ stratigraphic ranges, and consider a sampled tree
describing the phylogenetic relationship of these stratigraphic ranges. 
We define a sampled-ancestor-stratigraphic range as before. 
Let $j$ stratigraphic ranges be sampled-ancestor-stratigraphic ranges, the remaining $n-j$ stratigraphic ranges are  tip-stratigraphic ranges. 
The sampled tree has asymmetric branching times $x_1,\ldots,x_{w}$, symmetric branching times $x_{w+1},\ldots,x_{n-j-1}$, and a time of origin $x_0$. 
For derivations only, we again consider the oldest and youngest fossil of each stratigraphic range as explicit nodes that subdivide branches in the sampled tree. 
For convenience, as before, we count sampled nodes that represent a stratigraphic range consisting of a single fossil  (\textit{i.e.}, $o_i=y_i$) twice, as well as counting zero-length branches that begin and end at these sampled nodes. 
The sampled tree then consists of the following nodes:

\begin{itemize}
\item $w$ degree-three nodes, with the asymmetric branching times at $x_1,\ldots,x_{w}$, 
\item $n-j-1-w$ degree-three nodes, with the symmetric branching times at $x_{w+1},\ldots,x_{n-j-1}$, 
\item $n$ degree-two nodes, at the time of the oldest fossils $o_1,\ldots,o_n$,
\item $j$ degree-two nodes, at the sampled-ancestor-stratigraphic range  times $y_i$, with  $i=n-j+1,\ldots,n$,
\item $n-j$ degree-one nodes (tips), at the tip-stratigraphic range times $y_i$, $i=1,\ldots,n-j$, 
and
\item one degree-one node, the origin of the tree at time $x_0$.
\end{itemize}
In total, the sampled tree has $3n-j-1$ branches (also counting the initial branch beginning at the time of origin).  

As before, we define a set  $I$ consisting of stratigraphic ranges that have their most recent sampled ancestors on a straight line in the graphical representation.

\begin{thm} 
\label{ThmMixed}
In the case of potential incomplete sampling, the probability density of the oriented sampled tree $\mathcal{T}^o_s$ with $w$ asymmetric branching events, under mixed speciation is,
\begin{multline}\label{EqnMixed}
f[\mathcal{T}^o_s  \mid \lambda, \beta, \lambda_a, \mu, \psi, \rho, x_{0} ] = \\ (1-\beta)^w \beta^{n-j-1-w} \frac{\psi^k \rho^l \lambda^{n-j-1}}{ 1-p(x_0)} \prod_{i=1}^{3n-j-1}    \hat{q}(B_{i}) \prod_{i=1}^{n-j-l} p(y_i) \tsrev{ \prod_{i \in I} \Bigl(1 - \frac {q(o_i)} {\widetilde{q}(o_i)} \frac{ \widetilde{q}(y_{a(i)})} {q(y_{a(i)})}\Bigr)}, 
\end{multline}
where $p(t)$, $c_1$, $c_2$, $q(t)$ are defined as in Theorem~\ref{ThmExtendedStickTree},  and the contribution of branch $B_i$ with start time $s_i$ and end time $e_i$ is,
\begin{eqnarray*}
\hat{q}(B_i)=
\begin{cases}
\frac{\widetilde q(s_i)}{\widetilde q(e_i)} \quad \text{if branch } i \text{ is a stratigraphic range-branch,} \\ 
  \frac{q(s_i)}{q(e_i)} \quad \text{else,}
 \end{cases}
 \end{eqnarray*}
with 
\begin{equation*}
\widetilde q(t) := e^{-(\lambda_a+\beta (\lambda + \mu + \psi))t}
 ( \widetilde q_{asym}(t) )^{(1-\beta)}.
\end{equation*}
\end{thm}

\begin{proof}
Note that the probability densities $p(t)$ and $q(t)$ are the same as in the asymmetric case, as they do not depend on $\beta$ and $\lambda_a$. 
For $p(t)$, we note that the type of branching event does not influence the probability density of not sampling any descendants and only the total rate $\lambda$ will contribute to the expression for $p(t)$. 
The possibility of having a speciation event without branching does not influence the probability density of not sampling any descendants either. 
The equation for $q(t)$ also does not depend on the types of branching events that may have happened along the branch (\textit{i.e.}, asymmetric or symmetric speciation), because in both cases one lineage must not have been sampled and the other must have given rise to the observed tree.  
The possibility of having anagenetic speciation events along a lineage does not influence $q(t)$ either, because anagenetic speciation does not change the sampled tree.

The probability density of an individual associated with stratigraphic range $i$ at time $t \in [o_i,y_i)$ producing an stratigraphic range as observed within $(t,y_i]$ is described by the differential equation,
$$\frac{d}{dt}  \widetilde Q(t) = - (\lambda_a + \lambda + \mu + \psi)\widetilde Q(t) + (1-\beta) \lambda \widetilde Q(t) p(t).$$ 
Here, given that we know that the whole branch belongs to the same species we can eliminate the possibility of anagenetic or symmetric speciation events, along with sampling or death events. 
There may still be asymmetric speciation events along this branch, but the descendant species must not have been sampled, which is accounted for by the second term. 
For stratigraphic range $i$, the initial condition is $\widetilde Q(y_i) = c$ with $c=\psi p(y_i)$ if $y_i>0$ and $i$ is a tip-stratigraphic range, $c=Q(y_i)$ if $y_i>0$ and $i$ is a sampled-ancestor-stratigraphic range, and $c=\rho$ for $y_i=0$.
Plugging the expression for $\widetilde q(t)$ into the differential equation proves that this is a solution with $\widetilde q(0)=1$.
 Thus, $\widetilde Q(t) = \frac{\widetilde q(t)}{\widetilde q(y_i)}c$, and the contribution of stratigraphic range $i$ and its descendants to the probability density of the sampled tree is $\frac{\widetilde q(o_i)}{\widetilde q(y_i)} c$.

In other words, we can write the probability density of the oriented sampled tree with fossil samples partitioned into stratigraphic ranges by multiplying the contribution of all branches ($\hat{q}(B_{i})$) together with the initial values at the tips, the speciation rates, the fossilization rates, and the term for conditioning on a sample, as before. 

\tsrev{The right-most product  in Eq.\,\eqref{EqnMixed} accounts for the required unobserved speciation events prior to the stratigraphic ranges in $I$. Again, for each stratigraphic range $i \in I$, we multiply by $\frac{q(o_i)}{q(y_{a(i)})}$ as we do not want to continue assuming that the interval within $[o_i,y_{a(i)}]$ is associated with an arbitrary number of unobserved events. Then we take the difference of (i) the probability that in this interval any number of unobserved speciation events that change a species along that lineage happened ($\frac{q(y_{a(i)})}{q(o_i)}$) and (ii) the probability that no unobserved speciation   event that change a species along that lineage happened  ($ \frac {\widetilde{q}(y_{a(i)})} {\widetilde{q}(o_i)}$). Simplifying $\frac{q(o_i)}{q(y_{a(i)})}   \left(\frac{q(y_{a(i)})}{q(o_i)} - \frac {\widetilde{q}(y_{a(i)})} {\widetilde{q}(o_i)} \right)$ yields the expression stated in the theorem. Note that as an alternative proof for the right-most product, we could have integrated over $t_i$ as in Theorem \ref{ThmStickTree}.}

\end{proof}

\begin{cor} 
\label{ThmMixedTreeComplete}
Under mixed speciation, in the case of guaranteed complete sampling, the probability density of the oriented sampled tree, $\mathcal{T}^o_s$, is
\begin{multline*}
f[\mathcal{T}^o_s \mid \lambda, \beta, \lambda_a, \psi, x_{0} ] =  (1-\beta)^w\beta^{n-j-1-w} \psi^k \lambda^{n -j -1} 
\prod_{i=0}^{n-j-1}  e^{-(\lambda+\psi)x_i} \tsrev{\prod_{i=0}^{n}  e^{-\lambda_a(o_i-y_i)}} \\ \prod_{i \in I}  \tsrev{(1-  e^{-\lambda_a(y_{a(i)}-o_i)})}.
\end{multline*}
\end{cor}
\tsrev{Note that the second to last product accounts for no  anagenetic speciation events within stratigraphic ranges, and the last product accounts for at least one anagenetic speciation event occourring between $y_{a(i)}$ and $o_i$.}

Setting $\beta$ to one and $\lambda_a$ to zero (that is, allowing for only symmetric speciation events) one can obtain the corresponding probability densities for the FBD process with symmetric speciation only. 

\begin{cor}
In the case of potential incomplete sampling, the probability density of the oriented sampled tree $\mathcal{T}^o_s$ under symmetric speciation is,
$$f[\mathcal{T}^o_s \mid \lambda, \mu, \psi, \rho, x_{0} ] =  \frac{\psi^k \rho^l \lambda^{n-j-1}}{ (1-p(x_0))} \prod_{i=1}^{3n-j-1}    \hat{q}_{sym}(B_{i}) \prod_{i=1}^{n-j} p(y_i) \tsrev{ \prod_{i \in I}  \Bigl(1 - \frac {q(o_i)} {\widetilde{q}_{sym}(o_i)} \frac{ \widetilde{q}_{sym}(y_{a(i)})} {q(y_{a(i)})}\Bigr)},
$$
where $p(t)$, $c_1$, $c_2$, and $q(t)$ are defined as in Theorem~\ref{ThmExtendedStickTree}, and
\begin{eqnarray*}
\hat{q}_{sym}(B_i)=
\begin{cases}
\frac{\widetilde q_{sym}(s_i)}{\widetilde q_{sym}(e_i)} \quad \text{if branch } i \text{ is a stratigraphic range-branch,} \\ 
  \frac{q(s_i)}{q(e_i)} \quad \text{else,}
 \end{cases}
 \end{eqnarray*}
with  $\widetilde {q}_{sym}(t) := \widetilde q (t|\lambda, \mu, \psi, \rho, \beta=1, \lambda_a=0) = e^{-(\lambda + \mu + \psi)t}$. 
\end{cor}

As expected the expressions for densities $\widetilde q_{asym}(t)$ and $\widetilde q_{sym}(t)$ can be obtained from $ \widetilde q(t)$ by setting $\lambda_a$ to zero and $\beta$ to the extreme values, that is, $\widetilde q_{asym}(t|\lambda, \mu, \psi, \rho)  =  \widetilde q (t|\lambda, \mu, \psi, \rho, \beta=0, \lambda_a=0)$ and $\widetilde q_{sym}(t|\lambda, \mu, \psi, \rho)  =  \widetilde q (t|\lambda, \mu, \psi, \rho, \beta=1, \lambda_a=0)$.

\tsrev{The probability densities derived here can also be used for extinct clades by setting $\rho=1$, acknowledging the fact that we would include all extant species but there are none.}

\section{Marginalizing over the number of fossils within a stratigraphic range}
There may be a degree of uncertainty associated with the number of fossil specimens that were sampled throughout the stratigraphic range of a given species. 
In many cases, more effort has gone into researching the age of the  oldest ($o_i$) and youngest ($y_i$) fossils (the first and last appearances) of a given species, and it is rare that fossils have been sampled within the stratigraphic range with a  constant rate $\psi$. 
Thus, we now derive an expression for the probability density of a tree, given the oldest and youngest fossils of each sampled species, marginalizing over the number of fossils within this range.
In other words, we integrate over the number of fossil samples, $k$, for the probability densities derived above.

Let $\kappa'$ be the total number of sampled fossils that represent the start and end times of a stratigraphic range. 
If a stratigraphic range is represented by a single fossil then this fossil only contributes one towards $\kappa'$. 
Let $\kappa$ be the total number of sampled fossils that are found within any given stratigraphic range. In our example in Figure\ \ref{FigTree} we have $\kappa=3$.
Note that $k=\kappa'+\kappa$.
Let the sum of all stratigraphic range lengths be $L_s = \sum_{i=1}^n o_i-y_i$.

The symbol $\mathcal{T}$ denotes an extended sampled tree $\mathcal{T}_e$ (under asymmetric speciation) or a set of  extended sampled stratigraphic ranges $D$ (under asymmetric speciation) or  a sampled tree $\mathcal{T}_s$ (under asymmetric, symmetric, or mixed speciation). 
Further, $\mathcal{T}$ may be oriented or labeled. 
Let $\mathcal{T}_r$ be $\mathcal{T}$, ignoring the $\kappa$ fossils sampled within stratigraphic ranges. 
We further denote the parameters of the FBD model $(\lambda, \beta, \lambda_a, \mu, \psi, \rho)$ with $\eta$.

\begin{thm} \label{ThmIntegralk}
Both in the case of potential incomplete sampling and guaranteed complete sampling, the probability density of $\mathcal{T}_r$ is,
\begin{eqnarray*} 
f[\mathcal{T}_r \mid \eta,x_0] &=&   \psi^{-\kappa} f[\mathcal{T} \mid \eta, x_0] e^{\psi L_s}.
\end{eqnarray*}
Note that $\kappa$ is unknown, however $\psi^{-\kappa}$ cancels out with $\psi^{\kappa}$ in function $f[\mathcal{T}]$, meaning $f[\mathcal{T}_r ]$ does not depend on $\kappa$ \tsrev{while it depends on $\kappa'$}.
\end{thm}

\begin{proof}

Note that $\mathcal{T}$ can be obtained from $\mathcal{T}_r$ by adding the times, $\tau_1,\ldots, \tau_\kappa$, of the $\kappa$ fossils sampled within the stratigraphic ranges. 
Then $f[\mathcal{T}_r, \tau_1,\ldots, \tau_\kappa \mid \eta, x_0] = f[\mathcal{T} \mid \eta, x_0]$ and can be written as:
$$f[\mathcal{T}_r, \tau_1,\ldots, \tau_\kappa \mid \eta, x_0] = \psi^\kappa H,$$ with $H:=\psi^{-\kappa}   f[\mathcal{T} \mid \eta, x_0]$. 

From Theorem \ref{ThmExtendedStickTree} (resp.\ Corollary \ref{CorExtendedStickTreeLabeled}) for  extended oriented (resp.\ labeled) sampled trees under asymmetric speciation, 
Corollary \ref{CorExtendedStickTreeCollapse} for  extended sampled stratigraphic ranges under asymmetric speciation,
and Theorem \ref{ThmMixed} for oriented sampled trees under mixed speciation (and thus in particular under asymmetric or symmetric speciation),  
we observe that $H$ is independent of $\kappa$ and $\tau=(\tau_1,\ldots, \tau_\kappa)$ under potential incomplete sampling, while $H$ depends on the value $\kappa'$.

In the case of guaranteed complete sampling (\textit{i.e.}, $\mu=0$ and $\rho=1$),  
Theorem \ref{ThmExtendedStickTreeComplete}, Corollary \ref{CorExtendedStickTreeLabeled}, \ref{CorExtendedStickTreeCollapse}, and Corollary~\ref{ThmMixedTreeComplete} show that $H$ is independent of $\kappa$ and $\tau$, and again $H$ depends on the value $\kappa'$.

We now want to integrate over all $\tau$ to obtain $f[\mathcal{T}_r, \kappa \mid \eta, x_0 ]$, and then sum over all $\kappa$, to eliminate $\kappa$.
Note that each \tsrev{of the $\kappa$} fossil may be placed anywhere along the stratigraphic ranges with total length $L_s$.

Thus,
\begin{eqnarray} 
f[\mathcal{T}_r, \kappa \mid \eta, x_0] = \int_\tau f[\mathcal{T}_r,\tau_1,\ldots, \tau_\kappa \mid \eta, x_0 ] d \tau &=&  e^{\psi L_s} H \int_\tau \psi^{\kappa} e^{-\psi L_s} d \tau \notag \\
&=& e^{\psi L_s} H \frac{ (\psi L_s)^{\kappa} e^{-\psi L_s}} {\kappa !} . \label{MarginalHelp}
\end{eqnarray}
In the last equation we employed the fact that $\psi^{\kappa} e^{-\psi L_s}$ is the probability density of a realization of a Poisson process with $\kappa$ events over time period $L_s$.
Summing over all $\kappa$ leads to,
\begin{equation*} 
f[\mathcal{T}_r \mid \eta, x_0] = \sum_{\kappa=0}^\infty f[\mathcal{T}_r, \kappa \mid \eta, x_0]=  e^{\psi L_s} H,
\end{equation*}
which establishes the theorem.
\end{proof}

\begin{rem} \label{RemAsymMargi}
Under asymmetric speciation with guaranteed complete sampling for oriented trees, based on Theorem \ref{ThmExtendedStickTreeComplete}, with $L$ being the sum of all branch lengths, \tsrev{Theorem \ref{ThmIntegralk} simplifies to},
$$H =  \psi^{\kappa'} \lambda^{n-1}  \prod_{i=1}^n e^{-(\lambda+\psi)b_i}= \psi^{\kappa'} \lambda^{n-1}  e^{-(\lambda+\psi)L}.$$
Thus,
$$f[\mathcal{T}_r \mid \lambda, \psi, x_{0}] = e^{\psi L_s} H = \psi^{\kappa'} \lambda^{n-1}  e^{-\lambda L} e^{-\psi (L-L_s)}.$$ 

This probability density can also be proven in a direct way. 
The term $\psi^{\kappa'}$ is  the probability density of the fossils at the start and end of a stratigraphic range being sampled. 
The term $ \lambda^{n-1} $ is the rate for the $n-1$ branching events. The probability that no branching events happened along any of the branches is $e^{-\lambda L}$.
The probability that no sampling event happened along any of the branches outside the stratigraphic ranges is $e^{-\psi (L-L_s)}$.
\end{rem}

\begin{rem} \label{RemMixedMargi}
Under mixed speciation with guaranteed complete sampling for oriented trees, based on Corollary \ref{ThmMixedTreeComplete}, with $L$ being the sum of all branch lengths, \tsrev{Theorem \ref{ThmIntegralk} simplifies to},
\begin{multline}
H = \\ 
 \psi^{\kappa'} \lambda^{n- j- 1} (1-\beta)^w \beta^{n-j-1-w}  
e^{-(\lambda+\psi)L} \tsrev{e^{-\lambda_a L_s}\prod_{i \in I} (1-  e^{-\lambda_a(y_{a(i)}-o_i)})}.
\end{multline}
Thus,
\begin{multline*}
f[\mathcal{T}_r \mid \lambda, \beta, \lambda_a, \psi, x_{0}] = e^{\psi L_s} H = \\
\psi^{\kappa'} \lambda^{n-j-1} (1-\beta)^w \beta^{n-j-1-w}  
e^{-\lambda L} e^{-\psi (L-L_s)} \tsrev{e^{-\lambda_a L_s}\prod_{i \in I} (1-  e^{-\lambda_a(y_{a(i)}-o_i)})}.
\end{multline*} 
This probability density can also be proven in a direct way. 
The term $\psi^{\kappa'}$ is  the probability density of the fossils at the start and end of a stratigraphic range being sampled. 
The term $ \lambda^{n-j-1} $ is the rate for the $n-j-1$ branching events, $w$ of which are asymmetric, while the remaining $n-j-1-w$ are symmetric, 
which is accounted for by $(1-\beta)^w \beta^{n-j-1-w}$.  
The probability that no branching events happen along any of the branches is $e^{-\lambda L}$.
The probability that no sampling events happen along any of the branches outside the stratigraphic ranges is $e^{-\psi (L-L_s)}$. \tsrev{The probability that no anagenetic speciation events happen along the stratigraphic ranges is $e^{-\lambda_a L_s}$.}
\tsrev{The term $1-e^{-\lambda_a(y_{a(i)}-o_i)}$} accounts for unobserved anagenetic speciation events that \tsrev{must} have taken place between pairs of \tsrev{ancestor}-descendant stratigraphic ranges that lie along the same line.
\end{rem}

\section{Marginalizing over the number of fossils within a stratigraphic interval}
Instead of recording the age of the oldest and youngest fossils precisely (see previous section), some datasets may only record whether a fossil species was present or not within a given stratigraphic interval spanning the time interval $[x,y]$. 
Thus, a branch of a sampled tree within the time interval $[x,y]$ has either one or no fossil ``samples'' assigned to it, meaning only the presence or absence of a species is recorded. 
In other words, an assignment of one means that at least one fossil specimen of a particular species was found, but in fact any number \tsrev{$k_{x,y}$} fossil specimens may have been found within that  interval.

We will now derive equations accounting for only recording presence\,/\,absence of fossil specimens for a species rather than the exact number  of fossil specimens  \tsrev{$k_{x,y}$} in each stratigraphic interval.

As in the last section, the symbol $\mathcal{T}$ denotes an extended sampled tree $\mathcal{T}_e$ (under asymmetric speciation) or a set of  extended sampled stratigraphic ranges $D$ (under asymmetric speciation) or  a sampled tree $\mathcal{T}_s$ (under asymmetric, symmetric, or mixed speciation). 
Further, $\mathcal{T}$ may be oriented or labeled. A branch in $\mathcal{T}$ connects speciation nodes and/or tip nodes; fossil samples do not induce new branches.
Now we subdivide all branches in $\mathcal{T}$ into sub-branches, the time points for the start and end of the sub-branches are the start and end points of stratigraphic intervals.

\tsrev{Since we do not know the timing for the first and last fossil ($o_i$ and $y_i$) for each stratigraphic range, we have to estimate it. We suggest two options. Either we numerically integrate over $o_i$ and $y_i$ using MCMC methods.
Alternatively, we make an approximation assuming that if a fossil is found in a particular time interval, the corresponding species existed throughout that time  interval, meaning $o_i$ (resp. $y_i$) are the start (resp. end) times of the stratigraphic intervals where a species was found first (resp. last). An exception is that if the fossil is ancestral (resp. descendant) of a speciation node within that interval, then the speciation node is the new time $y_i$ (resp. $o_i$).}

Let $k_\SB$ be the \tsrev{(unknown)} number of fossil specimens along sub-branch $\SB$. We set  $\kappa_{\SB}=1$ if $k_\SB >0$, and $\kappa_{\SB}=0$ otherwise, meaning $\kappa_\SB$ indicates the presence\,/\,absence of a species.
Let $\mathcal{T}_l$ be $\mathcal{T}$ \tsrev{using the information on $\kappa_{\SB}$ instead of $k_\SB$ and using the potentially altered $o_i,y_i$}, with $l$ referring to intervals in the stratigraphic record. 
Let $L_\SB$ be the length of sub-branch $\SB$.

\begin{thm} \label{ThmIntegralkLayers}
Both in the case of potential incomplete sampling and guaranteed complete sampling, the probability density of $\mathcal{T}_l$ is,
\begin{eqnarray*} 
f[\mathcal{T}_l \mid \eta, x_0] &=&  \psi^{-k} f[\mathcal{T} \mid \eta, x_0]  \prod_{\SB: \kappa_{\SB}=1} e^{\psi L_{\SB}} (1-e^{-\psi L_{\SB}}).
\end{eqnarray*}
Note that $k$ is unknown, however $\psi^{-k}$ cancels out with $\psi^{k}$ in function $f[\mathcal{T}]$, meaning $f[\mathcal{T}_l ]$ does not depend on $k$.
\end{thm}
\begin{proof}
The proof is very similar to the proof of Theorem~\ref{ThmIntegralk}.
First, note that $\mathcal{T}$ can be obtained from $\mathcal{T}_l$ by adding the times, $\tau_1,\ldots, \tau_k$, of the $k$ fossils sampled along sub-branches $\SB$ with $\kappa_\SB = 1$. 
Then $f[\mathcal{T}_l, \tau_1,\ldots, \tau_\kappa \mid \eta, x_0] = f[\mathcal{T} \mid \eta, x_0]$ and can be written as:
$$f[\mathcal{T}_l, \tau_1,\ldots, \tau_\kappa \mid \eta, x_0] = \psi^k \tsrev{\bar{H}},$$ with $\tsrev{\bar{H}}:=\psi^{-k}   f[\mathcal{T} \mid \eta, x_0]$. 
Analog to Theorem\,\ref{ThmIntegralk}, it can be shown that $\tsrev{\bar{H}}$ is independent of $k$ and $\tau=(\tau_1,\ldots, \tau_\kappa)$.

Analog to Eqn.\,(\ref{MarginalHelp}), we obtain, with $\tau_\SB$ being the $k_\SB$ fossil sampling times on sub-branch $\SB$,
\begin{eqnarray*} 
f[\mathcal{T}_l, \kappa \mid \eta, x_0]  &=& \tsrev{\bar{H}} \prod_{\SB: k_{\SB}>0} e^{\psi L_\SB} \int_{\tau_\SB} \psi^{k_\SB} e^{-\psi L_\SB} d \tau = \tsrev{\bar{H}} \prod_{\SB: k_{\SB}>0} e^{\psi L_{\SB}} \frac{ (\psi L_{\SB})^{k_{\SB}} e^{-\psi L_{\SB}}} {k_{\SB} !}.
\end{eqnarray*}
Since $\sum_{k_{\SB}=1}^\infty \frac{ (\psi L_{\SB})^{k_{\SB}} e^{-\psi L_{\SB}}} {k_{\SB} !} = 1-e^{-\psi L_{\SB}}$, summing over  $k_{\SB}$ for all $\SB$ leads to,
\begin{eqnarray*} 
f[\mathcal{T}_l \mid \eta, x_0] &=& \tsrev{\bar{H}} \prod_{\SB: \kappa_{\SB}=1} e^{\psi L_{\SB}} (1-e^{-\psi L_{\SB}}).
\end{eqnarray*}
\end{proof}

\begin{rem}
Under asymmetric speciation with guaranteed complete sampling for oriented trees, based on Theorem \ref{ThmExtendedStickTreeComplete} and \ref{ThmIntegralkLayers}, with $L$ being the sum of all branch lengths, we have, analog to Remark \ref{RemAsymMargi},
\begin{eqnarray*} 
f[\mathcal{T}_l \mid \lambda, \psi, x_0] &=&  \lambda^{n-1} e^{-\lambda L} e^{-\psi \sum_{\SB: \kappa_{\SB}=0} L_\SB}  \prod_{\SB: \kappa_{\SB}=1}  (1-e^{-\psi L_{\SB}}).
\end{eqnarray*}
Note that $e^{-\psi \sum_{\SB: \kappa_{\SB}=0} L_\SB}$ is the probability of no fossil samples along sub-branches with $\kappa_{\SB}=0$, and $ \prod_{\SB: \kappa_{\SB}=1}  (1-e^{-\psi L_{\SB}})$ is the probability of at least one fossil sample on each sub-branch with $\kappa_{\SB}=1$.

Under mixed speciation with guaranteed complete sampling for oriented trees, based on Corollary \ref{ThmMixedTreeComplete} and Theorem \ref{ThmIntegralkLayers}, with $L$ being the sum of all branch lengths, we have, analog to Remark \ref{RemMixedMargi},
\begin{multline*}
f[\mathcal{T}_l \mid \lambda, \beta, \lambda_a, \psi, x_0] =  \lambda^{n- j- 1} (1-\beta)^w\beta^{n-j-1-w} 
e^{-\lambda L} e^{-\psi \sum_{\SB: \kappa_{\SB}=0} L_\SB} \\ \prod_{\SB: \kappa_{\SB}=1} \tsrev{\left( e^{-\lambda_a L_\SB} (1-e^{-\psi L_{\SB}})\right)}
\tsrev{\prod_{i \in I} (1-  e^{-\lambda_a(y_{a(i)}-o_i)})}.
\end{multline*}
\end{rem}

\section{Discussion} 
Due to the lack of statistical models combining neontological data (such as molecular sequence data) and paleontological data (such as stratigraphic ranges), these data are typically not analyzed within a single framework. 
Here, we formulate the FBD model under different modes of speciation giving rise to phylogenies and stratigraphic ranges, allowing for incomplete sampling of extinct and extant species. 
We introduce novel macroevolutionary models where we explicitly model the mode of speciation through time in  a phylogenetic context. 
As part of these new models, we derived the probability density, $P[\mathcal{T}_s^o \mid x_{0},  \eta]$ (with $\eta=(\lambda, \beta, \lambda_a, \mu, \psi, \rho)$), 
of a phylogenetic tree (referred to as the sampled tree in the mathematical derivations) on fossil and extant species samples. 
Specifically, several samples may be assigned to a single species, yielding so-called stratigraphic ranges in the phylogenetic tree. Thus, our equations will allow for a coherent and  flexible analysis of paleontological and neontological data.

In particular, we derived the probability density of the phylogenetic tree under asymmetric (budding) speciation in Theorem \ref{ThmStickTree}, and for speciation being either asymmetric (budding), symmetric (bifurcating), or anagenetic, in Theorem \ref{ThmMixed}.
These phylogenetic trees may have sampled-ancestor-stratigraphic ranges, where the  entire stratigraphic range is an ancestor of a descendant sampled species. 
Treatment of sampled ancestors is computationally challenging, requiring novel operators (\textit{i.e.}, proposal mechanisms) in Bayesian analyses \cite{GavryushkinaEtAl2014, HeathEtAl2014, ZhangEtAl2016}.

In the case of asymmetric speciation, the extended stratigraphic range (meaning the species from its first sample to its extinction time) can never be a sampled-ancestor-stratigraphic range, as the extinction event terminates a lineage. 
Thus we explore the extended stratigraphic range further under asymmetric speciation.  Corollary \ref{CorExtendedStickTreeLabeled} states the probability density of the tree connecting all samples, while knowing the extinction times for each sampled species (\textit{i.e.}, the extended sampled tree). 
Corollary \ref{CorExtendedStickTreeCollapse}, taking advantage of the absence of sampled-ancestor-stratigraphic ranges, additionally integrates analytically over all tree topologies. 
Since under symmetric speciation, a speciation event coincides with the extinction of the ancestor species and thus in the extended sampled tree we may also have sampled-ancestor-stratigraphic ranges,  we do not explore the extended sampled tree further under mixed speciation.

We envision that Theorem \ref{ThmStickTree} and Theorem \ref{ThmMixed} will be useful when tree inference is based on molecular and morphological data (such as for mammals). These expressions consider oriented trees rather than labeled trees. Because the analytical solution for the probability density of labeled trees is not possible with our equations, this motivates adapting phylogenetic software to oriented trees in order to use the provided equations.
In the case of asymmetric speciation, if the inferred extinction times for each sampled species are of interest, then  Corollary \ref{CorExtendedStickTreeLabeled} for labeled trees is appropriate. 
Many fossil datasets contain only fossil occurrence times without any morphological or molecular information and the tree topology cannot be inferred, therefore, Corollary \ref{CorExtendedStickTreeCollapse} can be employed for such cases.
Silvestro et al.\ \cite{SilvestroEtAl2014} also considers fossil occurrences, however, the equations assume that at least one fossil per extinct species is sampled, while we allow for non-sampled extinct species. 

For some rare and well studied groups (\textit{e.g.}, terrestrial vertebrates, dinosaurs) we may know the number of specimens collected through time, a number required by the equations above.
In many circumstances, however, we  only have information about the first and last occurrence times of a species, but not necessarily how many fossils were sampled in between (\textit{e.g.}, many marine invertebrates). 
Thus we further provide Theorem \ref{ThmIntegralk} to integrate over the number of fossils within a stratigraphic range for any of the settings mentioned above.
In other circumstances, we may only have information about the presence or absence of a fossil species within a given stratigraphic  interval or layer, but not the total number of specimens of a particular species sampled within each interval. We take the presence\,/\,absence data into account in Theorem \ref{ThmIntegralkLayers}.

We focussed on a thorough mathematical treatment of the FBD  model under different modes of speciation in this paper. 
The results, namely the probability of a tree, $P[\mathcal{T} \mid x_{0},  \eta]$ with $\eta$ being the FBD model parameters $\eta=(\lambda, \beta, \lambda_a, \mu, \psi, \rho)$, will be crucial for inferring posterior distributions of trees and model parameters (such as speciation and extinction rates) based on molecular and morphological data from extant and fossil species, or fossil occurrence data. 
Denoting all those data with $data$, 
and summarizing all parameters from models of evolution for the molecular and morphological data with $\theta$, a Bayesian method aims to infer,
$$P[\mathcal{T},  x_{0},\eta, \theta \mid data] = P[data \mid \mathcal{T},\theta]  P[\mathcal{T} \mid  x_{0}, \eta] P[\eta, x_{0},\theta] / P[data].$$
Thus, with the FBD model under different modes of speciation implemented as prior densities in Bayesian inference tools, we can readily infer trees and parameters from both paleontological and neontological data. 
Bayesian inference under these models may also allow us to assess the common modes of speciation by estimating their rate parameters $\lambda, \beta, \lambda_a$. We use the word ``may'' in the previous sentence as it is not clear if all parameters can  be identified based on the available data. While we know that we can infer $\lambda$ from enough neontological and paleontological data, it will be exciting to explore \tsrev{the extent to which} we can estimate further details of the speciation process from these data, \textit{i.e.}, estimate $\beta$ and $\lambda_a$.
Further, our mathematical results open the door to performing species-tree/gene-tree inference incorporating several fossils through time from the same species by using the probability density from  Theorem \ref{ThmStickTree} and Theorem \ref{ThmMixed}  as a species-tree prior.

We explicitly model the mode of speciation  within a phylogenetic framework. Following the paleontological literature, we model asymmetric (budding), symmetric (bifurcating), and anagenetic speciation \cite{Foote1996-fn}. 
A branching event in the phylogeny gives rise to either an asymmetric or a symmetric speciation event. 
Thus, these two branching speciation modes reflect the divergence of populations. 
In particular, these two modes of speciation do not require statements about the morphological change along these lineages;
in fact, divergence may be  driven by molecular rather than morphological change as observed in cryptic species. 
While anagenetic speciation may also be driven by molecular or morphological change, this speciation mode can typically only be identified if morphological change has occurred along a lineage to distinguish younger members from earlier ancestral forms, and if the fossil record of a given group has been sufficiently densely sampled to observe this morphological change directly (e.g.~planktic forams).
Thus, from a phylogenetic perspective, one might argue that we only want to model speciation processes that do not rely on morphological change, \textit{i.e.}, we only model branching speciation and set $\lambda_a=0$. 
This, however, would require us to associate uncertainty with the stratigraphic range data in the sense of allowing for the possibility  that different stratigraphic ranges actually belong to the same species despite being morphologically  distinct.
In addition, this would also require us to allow for the possibility that a single stratigraphic range actually  represents multiple morphologically very similar species.  We leave it for future work to extend our model such that uncertainty in stratigraphic range data can be considered.

The asymmetric speciation scenario, and in particular Theorem \ref{ThmStickTree}, can also be useful in epidemiology for modeling transmission trees. 
For some patients from whom we take a pathogen sequence at time $s_i$, we may know that they have already been infected at time $o_i$. 
We may also know that they are
still infected at some time $y_i$ more recently than $s_i$. 
We assume for patient $i$ the ``stratigraphic'' range $(o_i,y_i)$ and
obviously, $y_i=s_i$ and/or $o_i=s_i$ is possible.  
Theorem \ref{ThmStickTree} provides the probability density of a sampled tree (\textit{i.e.}, sampled transmission tree). 
Furthermore, when applying the species-tree/gene-tree framework to pathogens, yielding to a transmission-tree/gene-tree framework, we can incorporate multiple sequences per patient to infer transmission trees.
We note that an oriented sampled transmission tree provides us with ancestor-descendant relationships between patients, however, an ancestor may not be the direct donor due to unsampled intermediate patients, unless we can assume guaranteed complete sampling. In the case of guaranteed complete sampling, our method may be considered as an alternative to classic methods on transmission tree reconstruction from genetic data (see \textit{e.g.}, \cite{Jombart2011Heredity, Hall2015PLoSCompBiol}). In the case of potential incomplete sampling,  Didelot et al.~\cite{Didelot2017MBE} recently  proposed a method for inferring transmission trees. Compared to our method, their method can provide donor-recipient pairs. However, their approach cannot integrate over unobserved patients analytically, but requires data augmentation. The latter can be very slow in the case of many unobserved patients.

In summary, more explicit treatment of paleontological and neontological data in a phylogenetic framework, as presented here, has the potential to yield more robust and accurate inferences of macroevolutionary parameters, such as phylogenetic relationships, divergence times, rates of diversification, and rates of fossil recovery. 
Furthermore, our mathematical results also offer potentially promising approaches for detailed analysis of pathogen sequence data from an epidemic.
We end by highlighting that our approaches focus on processes inducing trees via processes such as speciation, extinction, transmission, and recovery. For the future, it will be a great challenge to further incorporate reticulation processes such as hybridization, horizontal gene transfer, and recombination. 

\section{Acknowledgements}
We thank Walker Pett for pointing out a simplification of the expression for $\widetilde{q}_{asym}(t)$, and Daniele Silvestro for discussions on the mixed speciation model. \tsrev{We thank the editor and two anonymous reviewers for very constructive and helpful comments towards improving our manuscript.}

TS is supported in part by the European Research Council under the Seventh Framework Programme of the European Commission (PhyPD: grant agreement number 335529).
RCMW is supported by the ETH Z\"{u}rich Postdoctoral Fellowship and Marie Curie Actions for People COFUND program.
TAH is supported by National Science Foundation (USA) grants: DEB-1556615 and DEB-1556853. 
AJD is supported by the Marsden Fund, contract number UOA1611.

\bibliographystyle{abbrv}

\end{document}